\documentclass[11pt,a4paper]{article}
\usepackage[margin=3cm]{geometry}
\usepackage[utf8]{inputenc}
\usepackage[T1]{fontenc}
\usepackage{amsmath}
\usepackage{amsfonts}
\usepackage{amssymb}
\usepackage{amsthm}
\usepackage{graphicx}
\usepackage{lmodern}
\usepackage{hyperref}
\usepackage[dvipsnames]{xcolor}
\usepackage{bm}
\usepackage{stmaryrd}
\usepackage{framed} 
	\definecolor{shadecolor}{gray}{0.9}  
\usepackage{mathrsfs} 
\usepackage{textcomp}
\usepackage{authblk}

\usepackage{tikz}
	\usetikzlibrary{arrows,positioning}

\theoremstyle{plain}
	\newtheorem{theorem}{Theorem}
	\newtheorem{lemma}{Lemma}
	\newtheorem{proposition}{Proposition}
	\newtheorem{definition}{Definition}
\theoremstyle{definition}
	
\theoremstyle{remark}

\newcommand{\tot}{\leftrightarrow}

\newcommand{\To}{\Rightarrow}

\usepackage{bussproofs}
\usepackage{pifont}
\usepackage[style=numeric, backend=biber]{biblatex}
\usepackage{multicol}
\usepackage[british]{babel}
\usepackage{enumitem}
\usepackage[size=\footnotesize,loadshadowlibrary,shadow,color=orange!50]{todonotes}
\usetikzlibrary{babel}
\usetikzlibrary{cd}
\usetikzlibrary{arrows.meta, positioning}
\setlist{nolistsep,leftmargin=*}
\allowdisplaybreaks
\newcommand{\True}{1}
\newcommand{\False}{0}
\newcommand{\M}{\mathfrak{M}}
\newcommand{\vTrue}{\vDash^{\True}}
\newcommand{\vFalse}{\vDash^{\False}}
\newcommand{\Tot}{\Leftrightarrow}
\newcommand{\xmark}{\ding{55}}
\newcommand{\Th}{\mathrm{Th}}

\theoremstyle{plain}
\newtheorem{fact}{Fact}

\author{Igor Sedl\'{a}r}
\author{Pietro Vigiani}
\affil{Institute of Computer Science, Czech Academy of Sciences}
\title{Modal Extensions of Generalised Nelson Logics}

\begin{document}
\maketitle

\begin{abstract}
    Motivated by relevant epistemic logic, we extend Dunn's simple Kripke-style binary relational semantics for the semi-relevant logic $\mathbf{RM}$ to fit a modal extension of $\mathbf{RM}$. It is shown that the modal axioms and rules that need to be added to $\mathbf{RM}$ to obtain a sound and complete axiomatisation of its modal extension give an axiomatisation of modal extensions of all logics in the family of generalised Nelson logics, also studied by Dunn. 
\end{abstract}

\section{Introduction}\label{sec:Intro}


Recently, there has been renewed interest in relevant modal logics, as they provide a fine-grained model of reasoning about epistemic notions and draw distinctions that classical modal logic cannot draw \cite{SedlarVigiani2022,SedlarVigiani2024,Standefer2022a,StandeferMares2025}. However, relevant modal logics have substantially more complicated semantics than classical modal logics and they tend to be undecidable. 
The propositional relevant logic $\mathbf{RM}$ differs from mainstream relevant logics in these respects --- it has a simple relational semantics based on ordered structures, developed by Dunn \cite{Dunn1976-DUNAKS}, and it is decidable \cite[§29.3 (written by R.\ K.\ Meyer)]{AndersonBelnap1975}, \cite{Dunn1976-DUNAKS}. Modal extensions of $\mathbf{RM}$ can be equipped with generic Routley--Meyer semantics extended with modal accessibility relations \cite{Fuhrmann1990}, but Dunn-style semantics for them has not been studied, to the best of our knowledge. The aim of this work is to fill this gap. 

Besides belonging to the relevant logics family, $\mathbf{RM}$ belongs to the family of so-called \emph{generalised Nelson logics}, also studied by Dunn \cite{Dunn2000}. The main technical result of the present paper is a modular soundness and completeness theorem for modal extensions of generalised Nelson logics, including modal $\mathbf{RM}$, with respect to Dunn-style semantics. Specifically, we show that the same template of modal axioms and rules can be added to the propositional axioms for each generalised Nelson logic $\mathbf{L}$ to obtain a sound and complete axiomatisation of a modal extension of $\mathbf{L}$.

Our semantics for modal operators uses a single accessibility relation $R$ for both $\Box$ and $\Diamond$. Formulas of the form $\Box A$ are considered true (false) in $x$ iff $A$ is true in all (false in some) states $R$-accessible from $x$, and $\Diamond A$ is considered true (false) in $x$ iff $A$ is true in some (false in all) states $R$-accessible from $x$. 
Our main motivation for adopting this semantics is its simplicity. It also distinguishes our framework from the existing approaches to modal extensions of Nelson-type logics \cite{Drobyshevich2020,Gao2025,Goble2006,OdintsovWansing2004,Odintsov2010-ODIMLW-2,OdintsovWansing2017,RivieccioEtAl2017,Sherkhonov2008}, which adopt different truth and falsity conditions for modal operators, assume frame conditions other than those assumed there, or formulate their modal extensions on propositional bases different from those considered by us. 
A detailed discussion of related work is deferred to a separate section.

The paper is organised as follows. In Section \ref{sec:PropSetting}, we recall, with some modifications, Dunn's presentation of the family of propositional generalised Nelson logics, their semantics and axiomatisations. 
Section \ref{sec:Modal} introduces modal extensions of propositional generalised Nelson logics and establishes a soundness and completeness result for them.  
In Section \ref{sec:Related} we discuss how our framework differs from existing work in modal Nelson-type logics. 
Section \ref{sec:Conclusion} concludes the paper and lists some interesting problems left to future work.

\section{Generalised Nelson Logics}\label{sec:PropSetting}

In \cite{Dunn2000}, Dunn defines a family of non-classical propositional logics, so-called \emph{generalised Nelson logics}. Logics in the family are unified by the fact that (i) they all admit a natural Kripke--style semantics which is \emph{bilateral}, i.e.~it is formulated in terms of independent verification and falsification relations; (ii) they all contain \emph{strong implication} $\to$ and \emph{strong negation} $\neg$, thereby providing different treatments of constructivity along the lines of Nelson's logic; and (iii) positive intuitionistic logic $\mathbf{IPC}^+$ can be embedded in all of them.
In this section, we review the Hilbert systems and model theory for propositional generalised Nelson logics. Compared to Dunn's presentation, we make minor modifications and corrections, which we find useful in view of the modal extension of Section \ref{sec:Modal}.

\subsection{Hilbert systems}\label{sec:PropHilbert}

Starting from a Hilbert style presentation of the generalised Nelson family, Dunn axiomatises in \cite{Dunn2000} 12 different logics, for which we use $\mathbf{L}$ as a meta-variable:
$$
\mathbf{L} \in \{ \mathbf{B}_\mathbf{1}, \mathbf{B}_\mathbf{10}, \mathbf{B}_\mathbf{1c}, \mathbf{B}_\mathbf{10c}, 
\mathbf{N}_\mathbf{1}, \mathbf{N}_\mathbf{10}, \mathbf{N}_\mathbf{1c}, \mathbf{N}_\mathbf{10c}, 
\mathbf{BN}_\mathbf{1}, \mathbf{BN}_\mathbf{10}, \mathbf{BN}_\mathbf{1c}, \mathbf{BN}_\mathbf{10c} \}.
$$

\noindent Some of the above logics correspond to known systems, in particular $\mathbf{N_{1}}$ and $\mathbf{BN_{1}}$ are Nelson's three-valued logic of constructible falsity $\mathbf{N3}$ \cite{Nelson1949} and its four-valued version $\mathbf{N4}$ \cite{AlmukdadNelson1984}; $\mathbf{B_{10c}}$ is the semi-relevant logic $\mathbf{RM}$ \cite{Dunn1976a} mentioned in Section \ref{sec:Intro}; 
and $\mathbf{N_{10}}$ and $\mathbf{BN_{10}}$ are the lesser known logics $\mathbf{F^{**}}$ and $\mathbf{BN}$, developed in \cite{Slaney1989} with an application to the problem of indeterminacy of future contingents.\footnote{Note that $\mathbf{F^{**}}$ and $\mathbf{BN}$ are equivalent up to mutual embeddings to $\mathbf{N3}$ \cite{SpinksVeroff2010}, respectively  $\mathbf{N4}$ \cite{Spinks2018}.} 

As for the naming conventions, the following motivations roughly guide the inclusion of the letters/subscripts $\mathbf{B},\mathbf{N}$,$\mathbf{10}$,$\mathbf{c}$ in the name $\mathbf{L}$. $\mathbf{L}$ contains
$\mathbf{B}$ (for \emph{both}) only if $\mathbf{L}$ is paraconsistent, i.e.~$B$ is not derivable from $A \land \neg A$;
$\mathbf{N}$ (for \emph{neither}) only if $\mathbf{L}$ is paracomplete, i.e.~$A \lor \neg A$ is not derivable from $B$; 
$\mathbf{10}$ (for \emph{truth}/\emph{falsity}) only if $\mathbf{L}$ is contrapositive (cf.~axiom \eqref{L:toContrap});
$\mathbf{c}$ (for \emph{connected}) only if $\mathbf{L}$ contains a version of the chain axiom (cf.~axiom \eqref{L:Conn}).
To simplify the notation, in what follows we write for some string of subscripts $\mathbf{s}= \mathbf{ijk}$: (i)  $\mathbf{L} \in \mathbf{L_s}$ to denote that $\mathbf{L} \in \{\mathbf{B_s}, \mathbf{N_s}, \mathbf{BN_s}\}$; and (ii) $\mathbf{L} = \mathbf{L_{i(j)k}}$ to denote that $\mathbf{L} \in \{ \mathbf{L_{ik}},\mathbf{L_{ijk}}\}$.
All propositional logics are axiomatised over the same language $\mathcal{L}$, defined below, except for logics which include $\mathbf{B}$ and the subscript $\mathbf{10}$, which are axiomatised over $\mathcal{L}^t \supseteq \mathcal{L}$, also defined below.

\begin{definition}[Languages]
The propositional language $\mathcal{L}$ is the set of formulas generated by the following grammar:
\[ A :=  p  \mid \neg A \mid  A \land A \mid A \lor A \mid A \to A \]
where $p \in \mathit{Prop}$, where $\mathit{Prop}$ is a countably infinite set of propositional variables. We define $A \tot B$ as $(A \to B) \land (B \to A)$.
Let $\mathcal{L}^t$ be the extension of $\mathcal{L}$ with the nullary operator $t$ (the `Ackermann constant'). Let the modal languages $\mathcal{L}_{\Box \Diamond}$ and $\mathcal{L}^t_{\Box \Diamond}$ be the extensions of $\mathcal{L}$ and $\mathcal{L}^t$ with two unary operators $\Box$ and $\Diamond$ (cf. Section \ref{sec:Modal}).
\end{definition}

We will use $\mathcal{L}_{\mathbf{L}}$ to denote the specific language in which the logic $\mathbf{L}$ is formulated, be it $\mathcal{L}$ and $\mathcal{L}^t$ in this section or $\mathcal{L}_{\Box \Diamond}$ and $\mathcal{L}^t_{\Box \Diamond}$ in Section \ref{sec:Modal}.

To provide a way of understanding at once the many Hilbert systems for generalised Nelson logics, it is crucial to emphasise the fact noted at the beginning of this section, namely that all generalised Nelson logics contain $\mathbf{IPC}^+$. Indeed, Dunn showed in \cite{Dunn2000} that in each $\mathbf{L}$ intuitionistic implication is expressible by means of different formulas, and denoted with \emph{elliptical implication} the formula template $A \To_{\mathbf{L}} B$, which is identical in each $\mathbf{L}$ to a formula which behaves like intuitionistic implication. 

\begin{definition}[Elliptical implication]\label{def:Elliptical}
    \emph{Elliptical implication} in a logic $\mathbf{L}$ is parametrically defined as follows:
    $$
    A \To_\mathbf{L}B := 
    \begin{cases}
        A \to B & \textit{if } \mathbf{L} \in \mathbf{L_{1(c)}}\\
        (A \land t) \to B & \textit{if } \mathbf{L} = \mathbf{B_{10(c)}}\\
        A \to (A \to B) & \textit{if } \mathbf{L} = \mathbf{N_{10(c)}}\\
        (A \land t) \to ((A \land t) \to B) & \textit{if } \mathbf{L} = \mathbf{BN_{10(c)}}.\\
    \end{cases}
    $$
\end{definition}

Note that in the above definition the Ackermann constant $t$ is used only in the definition of $\To_{\mathbf{L}}$ for $\mathbf{L} \in \{\mathbf{B_{10(c)}},\mathbf{BN_{10(c)}}\}$, which is why these are the only logics which are defined over the language $\mathcal{L}^t$.

Elliptical implication can be used to axiomatise generalised Nelson logics on the basis of $\mathbf{IPC}^+$. In fact, the axiomatisation of each $\mathbf{L}$ in \cite{Dunn2000} builds on the relevant logic $\mathbf{RWC}$, i.e.~the relevant logic $\mathbf{R}$ without the contraction axiom $(A \to (A \to B)) \to (A \to B)$ and the contraposition axiom $(A \to B) \to (\neg B \to \neg A)$. To the axioms of $\mathbf{RWC}$, Dunn supplements, for each generalised Nelson logic $\mathbf{L}$, two sorts of axioms: (i) axioms and rules of $\mathbf{IPC^+}$ (formulated using elliptical implication $\To_{\mathbf{L}}$) and (ii) specific axioms characterising the core principles of strong implication $\to$ and strong negation $\neg$ in all $\mathbf{L}$. The remaining axioms of $\mathbf{IPC}^+$ not included in (i) are derivable from the axioms of $\mathbf{RWC}$. 
As a result, all axiom systems for the generalised Nelson logics extend a basic common template $\mathbf{RWC}_{\To_{\mathbf{L}}}$, where we stress again that $\To_{\mathbf{L}}$ is a formula template which denotes different elliptical implications in each system. 

\begin{definition}[Basic axiom system]\label{def:BasicAx}
    For all $\mathbf{L}$, the Hilbert system $\mathbf{RWC}_{\To_\mathbf{L}}$ is the smallest set of formulas of $\mathcal{L}_{\mathbf{L}}$ that (i) contains all formulas and formula templates of the form
    \begin{gather}
A \to A \label{L:id}\\
(A \to B) \to ((C \to A) \to (C \to B)) \label{L:Trans}\\
(A \to (B \to C)) \to (B \to (A \to C)) \label{L:Exch}\\
(A \land B) \to A \qquad (A \land B) \to B\label{L:andElim}\\
((A \to B) \land (A \to C)) \to (A \to (B \land C))\label{L:andIntro}\\
A \to (A \lor B) \qquad B \to (A \lor B)\label{L:orIntro}\\
((A \to C) \land (B \to C)) \to ((A \lor B) \to C)\label{L:orElim}\\
(A \land (B \lor C)) \to ((A \land B) \lor C)\label{L:Distribut}\\
A \tot \neg\neg A\label{L:DN}\\
\neg (A \lor B) \tot (\neg A \land \neg B) \qquad \neg (A \land B) \tot (\neg A \lor \neg B)\label{L:DeMorgan}\\
A \To_{\mathbf{L}} (B \To_{\mathbf{L}} (A \land B)) \label{L:ToAdj}\\
A \To_{\mathbf{L}} (B \To_{\mathbf{L}} A) \label{L:ToPosParadox}\\
(A \To_{\mathbf{L}} (B \To_{\mathbf{L}} C)) \To_{\mathbf{L}} ((A \To_{\mathbf{L}} B) \To_{\mathbf{L}} (A \To_{\mathbf{L}} C)) \label{L:ToSelfDist}\\
((A \To_{\mathbf{L}} B) \land (\neg B \To_{\mathbf{L}} \neg A)) \To_{\mathbf{L}} (A \to B) \label{L:toPresAntipres}\\
(A \land \neg B) \To_{\mathbf{L}} \neg(A \to B)  \label{L:to0CounterEg2}
    \end{gather}
\noindent and (ii) is closed under elliptical modus ponens:
\begin{gather}
    \dfrac{A \To_{\mathbf{L}} B, \ A}{B} \label{L:ToModusPonens}
\end{gather}
\end{definition}

On the basis of Definition \ref{def:BasicAx}, each generalised Nelson logic $\mathbf{L}$ is defined as an axiomatic extension of $\mathbf{RWC}_{\To_{\mathbf{L}}}$, where in every case one adds (i) name-conferring axioms for $\mathbf{L}$ \eqref{L:NoGluts}--\eqref{L:Conn}, i.e.~determining whether the logic is paraconsistent, paracomplete, contrapositive and whether it satisfies a version of the chain axiom; (ii) the Ackermann axiom \eqref{L:Ack} to logics $\mathbf{L} \in \{\mathbf{B_{10(c)}},\mathbf{BN_{10(c)}}\}$ which require the Ackermann constant $t$ to express elliptical implication $\To_{\mathbf{L}}$; and (iii) axioms \eqref{L:to0CounterEg}--\eqref{L:toContrap} determining the specific properties of strong implication $\to$ and strong negation $\neg$ in $\mathbf{L}$. These axiomatic variations are summarised in Table \ref{tab:DunnLogics}, where we remark that the symbol \xmark , corresponding to $\mathbf{B_{1}}$ and $\mathbf{B_{10}}$, indicates that these logics are not sound with respect to the semantics introduced in Section \ref{sec:PropModels}. We note that Dunn does not consider these logics to be part of the generalised Nelson logics family for the same reason.


\begin{definition}[Generalised Nelson logics]\label{def:PropAxiom}
\noindent Each \emph{generalised Nelson logic} $\mathbf{L}$ is the axiomatic extension of $\mathbf{RWC}_{\To_{\mathbf{L}}}$ with a subset of the following axioms, as specified by Table \ref{tab:DunnLogics}:
\begin{gather}
        (A \land \neg A) \To_{\mathbf{L}} B  \label{L:NoGluts}\\
        A \To_{\mathbf{L}} (B \lor \neg B) \label{L:NoGaps} \\
        (A \To_{\mathbf{L}} B) \lor (B \To_{\mathbf{L}} A) \label{L:Conn}\\
        A \tot (t \to A) \label{L:Ack}\\
        \neg(A \to B) \To_{\mathbf{L}} (A \land \neg B) \label{L:to0CounterEg}\\
        ((A \to B) \land \neg(A \to B)) \To_{\mathbf{L}} ( A \land \neg B)  \label{L:to0+NonTrue.1}\\
        (A \to B) \to (\neg B \to \neg A). \label{L:toContrap}
\end{gather}
\end{definition}

\begin{table}[!t]
    \centering
    \begin{tabular}{ll|l}
    \multicolumn{2}{l|}{\textbf{Logic}} & \textbf{Extension of $\mathbf{RWC}_{\To_{\mathbf{L}}}$ }  \\[1mm] \hline 
      $\mathbf{N_{ 1}}$ & $ = \  \mathbf{N3}$ \cite{Nelson1949} &  \eqref{L:NoGluts}, \eqref{L:to0CounterEg}  \\  \hline
      $\mathbf{N_{ 10}} $ & $ = \  \mathbf{F^{**}}$ \cite{Slaney1989}   & \eqref{L:NoGluts}, \eqref{L:Ack}, \eqref{L:to0CounterEg}, \eqref{L:toContrap} \\ \hline
      $\mathbf{B_{ 1}} $ & \ \xmark  & \eqref{L:NoGaps}, \eqref{L:to0+NonTrue.1}  \\ \hline
      $\mathbf{B_{ 10}}$ & \ \xmark  & \eqref{L:NoGaps}, \eqref{L:Ack}, \eqref{L:to0+NonTrue.1}, \eqref{L:toContrap} \\ \hline
      $\mathbf{BN_{ 1}} $ & $ = \  \mathbf{N4}$ \cite{AlmukdadNelson1984}  
      &\eqref{L:to0CounterEg} \\ \hline
      $\mathbf{BN_{ 10}} $ & $ = \  \mathbf{BN}$ \cite{Slaney1989} &\eqref{L:Ack}, \eqref{L:to0CounterEg}, \eqref{L:toContrap} \\  \hline
      $\mathbf{B_{10c}}$ & $= \ \mathbf{RM}$ \cite{Dunn1976a} & \eqref{L:NoGaps}, \eqref{L:Ack}, \eqref{L:to0+NonTrue.1}, \eqref{L:toContrap}, \eqref{L:Conn} \\ \hline
      $\mathbf{L_{1(0)c}}$ &  & $\mathbf{L}$, \eqref{L:Conn} 
    \end{tabular}
    \caption{Generalised Nelson logics.}
    \label{tab:DunnLogics}
\end{table}

\begin{definition}[Derivability]\label{def:Lderivable}
For a generalised Nelson logic $\mathbf{L}$, a formula $A \in \mathcal{L}_{\mathbf{L}}$ is \emph{$\mathbf{L}$-derivable} from a set of formulas $\Gamma \subseteq \mathcal{L}_{\mathbf{L}}$, written $\Gamma \vdash_{\mathbf{L}} A$, iff there is a sequence of formulas $\langle D_1, \ldots, D_n\rangle$ such that $D_n = A$ and for all $i \leq n$ we have one of the following: (i) $D_i \in \mathbf{L}$, (ii) $D_i \in \Gamma$, or (iii) there are $j, k < i$ such that $D_k = D_j \To_{\mathbf{L}} D_i$. 
Such a sequence $\langle D_1, \ldots, D_n\rangle$ will be called an \emph{$\mathbf{L}$-derivation} (of $D_n$) from $\Gamma$. 
We define $\Th_{\mathbf{L}}(\Gamma) = \{ A \mid \Gamma \vdash_{\mathbf{L}} A \}$.
\end{definition}

We will write $\Gamma, \Delta \vdash_{\mathbf{L}} A$ instead of $\Gamma \cup \Delta \vdash_{\mathbf{L}} A$. We will write $\Gamma, A \vdash_{\mathbf{L}} B$ instead of $\Gamma \cup \{ A \} \vdash_{\mathbf{L}} B$ and $\vdash_{\mathbf{L}} A$ instead of $\emptyset \vdash_{\mathbf{L}} A$. We omit the subscript $\mathbf{L}$ when it is clear from context.
The following obvious facts about $\vdash_{\mathbf{L}}$, which follow from the intuitionistic properties of $\To_{\mathbf{L}}$,  will be used throughout the paper.

\begin{fact}\label{fact:vdashL}
For all $\mathbf{L}$:
\begin{enumerate}
\item if $A \in \Gamma$, then $\Gamma \vdash_{\mathbf{L}} A$;
\item if $\Gamma \vdash_{\mathbf{L}} A$, then $\Gamma, \Delta \vdash_{\mathbf{L}} A$;
\item if $\Gamma \vdash_{\mathbf{L}} A$ and $\Gamma \vdash_{\mathbf{L}} A \To B$, then $\Gamma \vdash_{\mathbf{L}} B$;
\item if $\Gamma \vdash_{\mathbf{L}} A$ and $\Delta, A \vdash_{\mathbf{L}} B$, then $\Gamma, \Delta \vdash_{\mathbf{L}} B$.
\end{enumerate}
\end{fact}

It is noteworthy that the only rule of generalised Nelson logics is \emph{elliptical modus ponens} \eqref{L:ToModusPonens}, and that derivability is defined accordingly, with respect to $\To_{\mathbf{L}}$. However, standard \emph{modus ponens} for strong implication $\to$ is derivable from \eqref{L:ToModusPonens} thanks to the fact that strong implication in each $\mathbf{L}$ is indeed stronger than the corresponding elliptical implication $\To_\mathbf{L}$ (cf. Lemma \ref{lem:StrongToInt}). Therefore, in what follows we shall make use without mention of standard \emph{modus ponens}.
Moreover, confirmation that $\To_\mathbf{L}$ encodes intuitionistic implication in each $\mathbf{L}$ is offered by the provability of the standard deduction theorem, equating derivability of $B$ from $A$ in each $\mathbf{L}$ with the provability in $\mathbf{L}$ of the elliptical implication $A \To_{\mathbf{L}} B$ (cf. Lemma \ref{thm:DT}).

\begin{lemma}[Strong to elliptical]\label{lem:StrongToInt}
    For all $\mathbf{L}$, $\vdash_{\mathbf{L}} (A \to B) \To_{\mathbf{L}} (A \To_{\mathbf{L}} B)$.
\end{lemma}
\begin{proof}
    \cite[p.28]{Dunn2000}.
\end{proof}

\begin{lemma}[Local deduction theorem]\label{thm:DT}
For all $\mathbf{L}$, $\Gamma, A \vdash_{\mathbf{L}} B$ iff $\Gamma \vdash_{\mathbf{L}} A \To_{\mathbf{L}} B$.
\end{lemma}
\begin{proof}
    \cite[Theorem 19]{Dunn2000}.
\end{proof}

Before moving to the semantics, we note a first point of departure from \cite{Dunn2000} in the axiomatisation of generalised Nelson logics. Definition \ref{def:PropAxiom} specifies that, whenever $\mathbf{L} \in \mathbf{L_{1(0)c}}$, the formula $(A \To_{\mathbf{L}} B) \lor (B \To_{\mathbf{L}} A)$ is an axiom of $\mathbf{L}$. This is in contrast with \cite{Dunn2000}, where Dunn adds the chain axiom in the form $(A \to B) \lor (B \to A)$, which is easily seen to be stronger by Lemma \ref{lem:StrongToInt}. The stronger axiom is valid in all generalised Nelson logics except $\mathbf{BN_{10c}}$ (cf.~Proposition \ref{prop:DunnFix}).\footnote{See \cite[130--132]{Bimbó2024} as well.} Therefore, we emend \cite{Dunn2000} by using uniformly the weaker chain axiom $(A \To_{\mathbf{L}} B) \lor (B \To_{\mathbf{L}} A)$ whenever $\mathbf{L} \in \mathbf{L_{1(0)c}}$.

\subsection{Kripke models}\label{sec:PropModels}

Models for generalised Nelson logics are closely related to semantic structures used in a number of non-classical settings, including first and foremost  Nelson's logic of constructible falsity \cite{Nelson1949,AlmukdadNelson1984}, but also intuitionistic logic \cite{Kripke1963-KRISAO-2} and the logic of first-degree entailment $\mathbf{FDE}$ \cite{Dunn1976}. 

In particular, generalised Nelson models are akin to Kripke models for intuitionistic logic in that they are based on \emph{partially ordered} structures, consisting of information states, ordered by the amount of true and false information they contain. Additionally, the truth of elliptical implication is equated with preservation of truth from the antecedent to the consequent along the partial order. Next, generalised Nelson models are akin to \emph{bilateral models} for $\mathbf{FDE}$ in that they employ two independent satisfaction relations, one for verification and one for falsification, and such that negation flips truth and falsity. Such bilateral models allow us to naturally model the presence of \emph{gluts}, i.e.~a formula is both verified and falsified; and \emph{gaps}, i.e.~a formula is neither verified nor falsified. Finally, generalised Nelson models are akin to Nelson models in that the falsity of strong implication at a state is related to the existence of a local counterexample, i.e.~the state verifying the antecedent and falsifying the consequent.

Classes of generalised Nelson models for each $\mathbf{L}$ differ from one another in that they (i) impose different restrictions on the valuation $V$, (ii) impose different restrictions on the partial order, or (iii) adopt different verification and falsification conditions for strong implication.  Restrictions on $V$ have the effect of allowing gaps and/or gluts; moreover, they yield contraposition by sanctioning a sufficient condition for a state to falsify the Ackermann constant $t$ (for $\mathbf{L} \in \{\mathbf{B_{10(c)}},\mathbf{BN_{10(c)}}\}$ over $\mathcal{L}^t$). Finally, a restriction on $\leqslant$ is considered, which makes the chain axiom \eqref{L:Conn} valid. These semantic variations are summarised in Table \ref{tab:DunnModels}. 

\begin{definition}[Basic models]\label{def:BasicMod}
    A \emph{basic generalised Nelson model} is a structure \(\M = \langle K, \leqslant, V\rangle\)  where: 
    \begin{itemize}
    \item $\langle K, \leqslant\rangle$ is a partially ordered set;
    \item $V: \textit{Prop} \times K \to \mathcal{P}(\{1,0\})$ such that
\begin{gather}
    x \leqslant y \textit{ only if } V(p,x) \subseteq V(p,y) \label{eq:VHeredity}
\end{gather}
    \end{itemize}
    The class $\mathcal{N}$ consists of all basic generalised Nelson models $\mathfrak{M}$.
\end{definition}

On the basis of Definition \ref{def:BasicMod}, models $\mathfrak{M}_{\mathbf{L}}$ for each generalised Nelson logic $\mathbf{L}$ are specified as follows. By an abuse of notation, we write $(F)$ in the signature of $\mathfrak{M}_{\mathbf{L}}$ to indicate that $F$ may or may not be present, depending on how $\mathbf{L}$ is specified.

\begin{definition}[Models for $\mathbf{L}$]\label{def:GenNelsonModel}    
    \noindent A generalised Nelson \emph{$\mathbf{L}$-model}, or simply $\mathbf{L}$-model, is a structure $\M_{\mathbf{L}} = \langle K, \leqslant,(F), V\rangle$ (for $\mathbf{L}$ from Table \ref{tab:DunnLogics}) where:
    
\begin{itemize}
    \item $\langle K, \leqslant, V\rangle$ is a basic generalised Nelson model;
    \item $F \subseteq K$ (included in $\M_{\mathbf{L}}$ only where $\mathbf{L} \in \{\mathbf{B_{10(c)}},\mathbf{BN_{10(c)}}\}$) is such that
\begin{gather}
    x \leqslant y \text{ and } x \in F \textit{ only if } y \in F \label{eq:FHeredity}
\end{gather} 
    \item $\mathcal{M}_\mathbf{L}$ satisfies a subset of the following conditions, as detailed in Table \ref{tab:DunnModels}:
    \begin{gather}
        V(p,x) \neq \{1,0\} \label{eq:NoGluts} \\ 
        V(p,x) \neq \emptyset \label{eq:NoGaps}\\ 
        \exists p (V(p,x) = \{1,0\}) \Longrightarrow x \in F \label{eq:Ack}\\ 
         x \leqslant y \textit{, } x \leqslant z \Longrightarrow y \leqslant z \textit{ or } z \leqslant y \label{eq:Conn}
    \end{gather}
    \item $V$ is recursively extended to the whole language $\mathcal{L}$ ($\mathcal{L}^t$ for $\mathbf{L} \in \{\mathbf{B_{10(c)}}, \mathbf{BN_{10(c)}}\}$) as follows:
    
\medskip
\begin{tabular}{lcllcl}
   $x \vTrue p $ & iff & $1 \in V(p,x)$ & $x \vFalse p $ & iff & $0 \in V(p,x)$\\
   $x \vTrue t $ & iff & $ x \in K $ & $x \vFalse t $ & iff & $ x \in F$\\
   $x \vTrue \neg A  $ & iff & $ x \vFalse A$  & $x \vFalse \neg A  $ & iff & $ x \vTrue A$ \\
   $x \vTrue A \land B  $ & iff & $ x \vTrue A \text{ and } x \vTrue B$ & $x \vFalse A \land B  $ & iff & $ x \vFalse A \text{ or } x \vFalse B$ \\
   $x \vTrue A \lor B  $ & iff & $ x \vTrue A \text{ or } x \vTrue B $ & $x \vFalse A \lor B  $ & iff & $ x \vFalse A \text{ and } x \vFalse B$\\
   $x \vTrue A \to B$  & iff & $x \vTrue_{\mathbf{L}} A \to B$ & $x \vFalse A \to B$  & iff & $x \vFalse_{\mathbf{L}} A \to B$
\end{tabular}
\medskip

\noindent where $x \vTrue_{\mathbf{L}} A \to B$ and $x \vFalse_{\mathbf{L}} A \to B$ are specified according to Table \ref{tab:DunnModels}.
\end{itemize}
The class $\mathcal{N}_\mathbf{L}$ consists of all $\mathbf{L}$-models $\M_\mathbf{L}$.
\end{definition}

\begin{table}[!t]
    \centering
    \begin{tabular}{l|l|l}
    \textbf{L} & \textbf{Conditions on } $\mathcal{N}_\mathbf{L}$ & \textbf{$\to$ Verification/Falsification} \\[1mm] \hline 
      $\mathbf{N_{ 1}}$ & $\mathcal{N}$ + \eqref{eq:NoGluts}
      & $x \vTrue_{\mathbf{L}}    A \to B$ iff $\forall y \geqslant x(y \vTrue A \text{ only if } y \vTrue B)$  \\ 
      && $x \vFalse_{\mathbf{L}}  A \to B$ iff $x \vTrue A \land \neg B$ \\ \hline
      $\mathbf{N_{ 10}}$ & $\mathcal{N}$ + \eqref{eq:NoGluts}
      & $x \vTrue_{\mathbf{L}}  A \to B$ iff 
      $\begin{aligned}[t]
          &\forall y \geqslant x(y \vTrue A \text{ only if } y \vTrue B)\\[-2mm]
          &\forall y \geqslant x(y \vFalse B \text{ only if } y \vFalse A)
      \end{aligned}$
        \\
      && $x \vFalse_{\mathbf{L}}   A \to B$ iff $x \vTrue A \land \neg B$ \\  \hline
      $\mathbf{B_{ 1}} $ & $\mathcal{N}$ + \eqref{eq:NoGaps}
      & $x \vTrue_{\mathbf{L}}   A \to B$ iff $\forall y \geqslant x(y \vTrue A \text{ only if } y \vTrue B)$  \\
      && $x \vFalse_{\mathbf{L}}  A \to B$ iff $x \vTrue A \land \neg B$ or $x\not\vTrue_{\mathbf{L}} A \to B$ \\ \hline
      $\mathbf{B_{ 10}}$  & $\mathcal{N}$ + \eqref{eq:NoGaps}, \eqref{eq:Ack}
      & $x \vTrue_{\mathbf{L}}  A \to B$ iff 
      $\begin{aligned}[t]           
      &\forall y \geqslant x(y \vTrue A \text{ only if } y \vTrue B) \\[-2mm]   
      &\forall y \geqslant x(y \vFalse  B \text{ only if } y \vFalse  A)
      \end{aligned}$\\
      && $x \vFalse_{\mathbf{L}}  A \to B$ iff $x \vTrue A \land \neg B$ or $x\not\vTrue_{\mathbf{L}} A \to B$ \\ \hline
      $\mathbf{BN_{ 1}}$ & $\mathcal{N}$
      & $x \vTrue_{\mathbf{L}} A \to B$ iff $\forall y \geqslant x(y \vTrue A \text{ only if } y \vTrue B)$  \\
      && $x \vFalse_{\mathbf{L}}  A \to B$ iff $x \vTrue A \land \neg B$ \\ \hline
      $\mathbf{BN_{ 10}}$ & $\mathcal{N}$ + \eqref{eq:Ack}
      & $x \vTrue_{\mathbf{L}} A \to B$ iff $\begin{aligned}[t]           &\forall y \geqslant x(y \vTrue A \text{ only if } y \vTrue B) \\[-2mm]           &\forall y \geqslant x(y \vFalse  B \text{ only if } y \vFalse  A)       \end{aligned}$\\
      && $x \vFalse_{\mathbf{L}} A \to B$ iff $x \vTrue A \land \neg B$\\ \hline
      $ \mathbf{L_{1(0)c}}$ & $\mathcal{N}_{\mathbf{L_{1(0)}}}$ + \eqref{eq:Conn}
      & $x \vTrue_{\mathbf{L_{1(0)c}}}  A \to B$ iff $x \vTrue_{\mathbf{L_{1(0)}}} A \to B$\\
      && $x \vFalse_{\mathbf{L_{1(0)c}}}  A \to B$ iff $x \vFalse_{\mathbf{L_{1(0)}}} A \to B$ 
    \end{tabular}
\caption{Generalised Nelson models.}
    \label{tab:DunnModels}
\end{table}

It is worthwhile to note some points of departure with \cite{Dunn2000} in the presentation of models for generalised Nelson logics.

The first difference with respect to \cite{Dunn2000} concerns Condition \eqref{eq:Conn}. 
Whenever $\mathbf{L} \in \mathbf{L_{1(0)c}}$, we assume in $\M_{\mathbf{L}}$ that the partial order is a tree-like partial order, i.e.~the set of states greater than a given state is linearly ordered. In contrast, in \cite{Dunn2000} Dunn assumes the stronger condition that all states are connected, or linearly ordered, i.e.~$x \leqslant y$ or $y \leqslant x$ for all $x,y \in K$. In a propositional setting, either condition is sufficient for the purpose of soundness and completeness. However, modal extensions require one to use the weaker condition. Hence, we choose to work with Condition \eqref{eq:Conn} instead of linearity from the start.\footnote{A related point concerns the specific class of $\mathbf{B_{10c}}$-models, i.e.~models for $\mathbf{RM}$. The original models for $\mathbf{RM}$ of \cite{Dunn1976a} work with a linear order, where it is assumed without loss of generality that there is a base state $b$ such that all states are greater than $b$. In view of modal extensions, it is much easier to work without base states.} In order to keep our notation consistent with Dunn's, we keep using the letter $\mathbf{c}$ in the subscript, but the semantic condition it denotes is weaker than Dunn's `connectedness' (linearity). 

A second difference with \cite{Dunn2000} concerns the set $F$ of inconsistent states.\footnote{The name $F$ is reminiscent of Tedder's set of \emph{funky} states in the context of $\mathbf{RM}$ \cite{Tedder2022-TEDANO}.} As noted in Section \ref{sec:PropHilbert}, whenever $\mathbf{L} \in \{\mathbf{B_{10(c)}},\mathbf{BN_{10(c)}}\}$, the Ackermann constant $t$ is included in the language. For reasons explained below, this  is matched at the semantic level by the inclusion in the models of a distinguished set $F \subseteq K$ of \emph{inconsistent states}. In the present setting, it is explicitly assumed that a sufficient condition for a state to be an inconsistent state is that it admits a glut; and then in the canonical model it will be assumed that a state is inconsistent iff it falsifies $t$. Conversely, in \cite{Dunn2000} Dunn does not include the set $F$ in the models but builds the previous sufficiency condition directly in the semantics of $t$. Although the two strategies are again equivalent at the propositional level, only the first is adequate for the sake of modal extensions of the language. Hence, we chose again to adopt the more general strategy from the start.

To conclude the presentation of $\mathbf{L}$-models, we already observed that variability also concerns the definition of verification $\vTrue$ and falsification $\vFalse$ for strong implication. These variations again link generalised Nelson logics with a number of semantic frameworks. In particular, whenever $\mathbf{L} \in \mathbf{L_{1(c)}}$ the verification clause for $\to$ is the intuitionistic one, i.e.~preservation of truth from the antecedent to the consequent in all greater states. In addition to this, whenever $\mathbf{L} \in \mathbf{L_{10(c)}}$ it is required also that falsity is preserved from the consequent to the antecedent in all greater states; this addition enforces contraposition.\footnote{Indeed, although Nelson's logics $\mathbf{N_1} = \mathbf{N3}$ and $\mathbf{BN_1} = \mathbf{N4}$ are equivalent up to mutual embedding to $\mathbf{F^{**}} = \mathbf{N_{10}}$ and $\mathbf{BN} = \mathbf{BN_{10}}$ \cite{SpinksVeroff2010,Spinks2018}, respectively, the former two lack contraposition, while the latter two admit contraposition.} Moreover, whenever $\mathbf{L} \in \{\mathbf{N_{1(0c)}},\mathbf{BN_{1(0c)}}\}$, the falsification clause for $\to$ is the constructible counterexample one from Nelson's logic, i.e.~verification of the antecedent and falsification of the consequent at the given state. In addition to this, whenever $\mathbf{L} = \mathbf{B_{1(0c)}}$, an alternative to a constructible counterexample is admitted, i.e.~the failure of verification of $\to$. This so-called \emph{escape disjunct}, which is not in the spirit of constructivity, but is necessary if we want to forbid gaps, is characteristic of the falsification clause for $\mathbf{RM}$.

Validity and (local) consequence in each class of generalised Nelson models $\mathbf{L}$ is defined as usual. 

\begin{definition}
Let $\mathcal{N}_\mathbf{L}$ be a class of generalised Nelson models. A formula $A$ is a  \emph{consequence} of a set of formulas $\Gamma$ over $\mathcal{N}_\mathbf{L}$ (notation $\Gamma \vDash_{\mathcal{N}_\mathbf{L}} A$) iff, for all models in $\M \in \mathcal{N}_\mathbf{L}$ and all $x \in K$ of $\M$, we have
\[ x \vTrue B \text{ for all } B \in \Gamma \implies x \vTrue A.\]
In particular,  $A$ is \emph{valid in $\mathcal{N}_\mathbf{L}$} ($\vDash_{\mathcal{N}_\mathbf{L}} A)$ iff $\emptyset \vDash_{\mathcal{N}_\mathbf{L}} A$ and $A$ is \emph{valid in $\M$} iff $A$ is valid in $\{ \M \}$. Finally, we write $\Gamma \vDash_\mathfrak{M} A$ instead of $\Gamma \vDash_{\{\mathfrak{M}\}} A$.
\end{definition} 

Having introduced the main semantic ingredients, we turn now to the identification of three crucial properties of generalised Nelson models. These properties, expressed in Lemmas \ref{lem:Valuation}--\ref{lem:FSufficiency}, are important insofar as they streamline the proof of soundness. Each of them generalises to arbitrary formulas a property assumed for atomic formulas on specific classes of $\mathbf{L}$-models.  Lemma \ref{lem:Valuation} states that $\mathbf{N_{1(0c)}}$-models do not admit gluts and $\mathbf{B_{1(0c)}}$-models do not admit gaps, generalising Conditions \eqref{eq:NoGluts} and \eqref{eq:NoGaps}. Lemma \ref{lem:Heredity} states that verification and falsification are preserved along $\leqslant$ in all $\mathbf{L}$-models, generalising Condition \eqref{eq:VHeredity}. Finally, Lemma \ref{lem:FSufficiency} states that in all $\mathbf{L_{10(c)}}$-models (thus, containing the set of inconsistent states $F$) a sufficient condition for a state to be inconsistent is that it both verifies and falsifies some arbitrary formula, generalising Condition \eqref{eq:Ack}.

\begin{lemma}\label{lem:Valuation}
    For all $\mathbf{L}$-models $\M_\mathbf{L}$:
    \begin{enumerate}
        \item If $\mathbf{L} = \mathbf{N_{1(0c)}}$ then $x \vTrue A$ only if  $x \not\vFalse A$;
        \item If $\mathbf{L} = \mathbf{B_{1(0c)}}$, then $x \not\vTrue A$ only if  $x \vFalse A$.
    \end{enumerate}
\end{lemma}
\begin{proof}
    Both results are established by induction on the complexity of $A$, where the base case follows from \eqref{L:NoGluts}, respectively \eqref{L:NoGaps}. The only non-trivial induction step case is $A= B\to C$, where we reason as follows. 
    
    \noindent \emph{1.} Assume contraposively that $x \vFalse B \to C$, which implies $x \vTrue B$ and $x \vFalse C$. By induction hypothesis we infer that $x \not \vTrue C$, which implies by $x \leq x$ and  $x \vTrue B$ that $x \not\vTrue B \to C$.
    
    \noindent \emph{2.} It holds by definition of $x \vFalse_{\mathbf{B_1(0c)}} B \to C$.
\end{proof}

\begin{lemma}\label{lem:Heredity}
    For all $\mathbf{L}$-models $\M_\mathbf{L}$ such that $\mathbf{L} \neq \mathbf{B_{1(0)}}$:
    \begin{enumerate}
        \item $x \leqslant y, x \vTrue A \textit{ only if } y \vTrue A $;
        \item $x \leqslant y, x \vFalse A \textit{ only if } y \vFalse A$.
    \end{enumerate}
\end{lemma}
\begin{proof}
    By double induction on the complexity of $A$. The base case holds by Condition \eqref{eq:VHeredity}. The only non-trivial induction step case is $A= B \to C$, where we reason as follows. 

    \noindent \emph{1.} Assume $x \vTrue B \to C$ and $x \leqslant y$. If $\mathbf{L} \in \mathbf{L_{10(c)}}$, it suffices to show that (i) $z \vTrue B$ only if $z \vTrue C$ and (ii) $z \vFalse C$ only if $z \vFalse B$, for some arbitrary $z \geqslant y$. (i) By $x \vTrue B \to C$, $x \leqslant y \leqslant z$ and $z \vTrue B$ we conclude that $z \vTrue C$. (ii) By $x \vTrue B \to C$, $x \leqslant y \leqslant z$ and $z \vFalse C$ we conclude that $z \vFalse B$. If $\mathbf{L} \in \mathbf{L_{1(c)}}$, then we need to check only the above case (i). 

    \noindent \emph{2.} Assume $x \vFalse B \to C$ and $x \leqslant y$. If $\mathbf{L} \in \{\mathbf{N_{1(0c)}}, \mathbf{BN_{1(0c)}}\}$, then $x \vTrue B \land \neg C$, whence $y \vFalse B \land \neg C$ holds by straightforward application of the induction hypothesis. We conclude that $y \vFalse B \to C$. If $\mathbf{L} = \mathbf{B_{10c}}$, then either $x \vTrue B \land \neg C$, and we reason as above, or $x \not\vTrue B \to C$. The latter implies that either (i) there is $z \geqslant x$ such that $z \vTrue B$ and $z \not\vTrue C$ or (ii) there is $z \geqslant x$ such that $z \vFalse C$ and $z \not\vFalse B$. In both cases, by $x \leqslant y$ and $x \leqslant z$ we infer by \eqref{eq:Conn} that either $y \leqslant z$ or $z \leqslant y$. If $y \leqslant z$ and either (i) or (ii), we conclude that $y \not\vTrue B \to C$, whence $y \vFalse B \to C$. If $z \leqslant y$ and (i), then $z \not\vTrue C$ implies by Lemma \ref{lem:Valuation} (since $\mathbf{L} = \mathbf{B_{10c}}$) that $z \vFalse C$, which implies by $z \vTrue B$ and the induction hypothesis that $y \vTrue B \land \neg C$. We conclude that $y \vFalse B \to C$. If $z \leqslant y$ and (ii), we reason in a similar way. If $\mathbf{L} = \mathbf{B_{1c}}$, then we reason similarly to the case $\mathbf{L} = \mathbf{B_{10c}}$, except that we now need to check only case (i). 
\end{proof}

By reviewing the above proof, one may note the essential use of Condition \eqref{eq:Conn} in the case $\mathbf{L} = \mathbf{B_{1(0)c}}$. Indeed, Lemma \ref{lem:Heredity} fails for $\mathbf{L} = \mathbf{B_{1(0)}}$: models in  $\mathbf{B_{1(0)}}$ do not admit gaps, and so they need to feature what we called the escape disjunct in the falsification of strong implication. However, the escape disjunct seems to require to use Condition \eqref{eq:Conn} to establish that the falsity of strong implication is preserved along $\leqslant$. This is the reason why $\mathbf{B_1}$ and $\mathbf{B_{10}}$ are not officially included by Dunn in \cite{Dunn2000} in the family of generalised Nelson logics.

\begin{lemma}\label{lem:FSufficiency}
    For all $\mathbf{L}$-models $\M_\mathbf{L}$ such that $\mathbf{L} \in  \mathbf{L_{10(c)}}$, if $x \vTrue A \land \neg A$ for some $A \in \mathcal{L}_{\mathbf{L}}$, then $x \in F$.
\end{lemma}
\begin{proof}
    By induction on the complexity of $A$. The base case follows from Condition \eqref{eq:Ack}. The only non-trivial induction step case is $A =  B \to C$, where we reason as follows. Assume that $x \vTrue B \to C$ and $x \vFalse B \to C$. Hence, $x \vTrue B$ and $x \vFalse C$. By $x \vTrue B \to C$, $x \leqslant x$ and $x \vTrue B$ we infer that $x \vTrue C$. The latter and $x \vFalse C$ imply by the induction hypothesis that $x \in F$, which concludes the proof.  
\end{proof}

Finally, the following easy fact, proved in \cite{Dunn2000}, turns out useful in the soundness proof.

\begin{lemma}[Elliptical semantics]\label{lem:SemInt}
    For all $\mathbf{L}$-models $\M_{\mathbf{L}}$, $x \models^1 A \To_{\mathbf{L}} B$ iff $\forall y \geqslant x(y \vTrue A \text{ only if } y \vTrue B)$
\end{lemma}
\begin{proof}
    \cite[Lemma 30]{Dunn2000}.
\end{proof}

\subsection{Soundness and completeness}\label{sec:PropCompleteness}

In this section, we review Dunn's proof of soundness and completeness of all generalised Nelson logics with respect to the corresponding class of models. We state the result in some detail for two reasons. First, doing so  paves the way for the soundness and completeness result for modal extensions. Second, this allows us to rectify some of Dunn's results.

As usual, soundness is a matter of a routine check that, for all $\mathbf{L}$, the axioms of $\mathbf{L}$ are valid in $\mathcal{N}_\mathbf{L}$ and that elliptical modus ponens preserves validity in $\mathcal{N}_\mathbf{L}$. The proof involves many case distinctions among each $\mathbf{L}$, and it essentially relies on Lemmas \ref{lem:Heredity}--\ref{lem:FSufficiency}: for this reason, $\mathbf{B_{1}}$ and $\mathbf{B_{10}}$ are not sound, since Lemma \ref{lem:Heredity} does not hold.

\begin{theorem}[Soundness]\label{thm:Soundness}
    $\Gamma \vdash_{\mathbf{L}} A$ only if $\Gamma \models_{\mathcal{N}_\mathbf{L}} A$.
\end{theorem}
\begin{proof}
    By induction on the length of $\mathbf{L}$-proofs. We show the details only for the base case concerning \eqref{L:Conn} (assuming $\mathbf{L} \in \mathbf{L_{1(0)c}}$) and \eqref{L:Ack} (assuming $\mathbf{L} \in \{\mathbf{B_{10(c)}},\mathbf{BN_{10(c)}}\}$), which are the only ones differing from Dunn's argument in \cite[Theorem 22]{Dunn2000}.
    
    \noindent \eqref{L:Conn}. Assume by contradiction that there is $x$ such that $x \not\vTrue A \To_{\mathbf{L}} B$ and $x \not\vTrue B \To_{\mathbf{L}} A$. Hence, by Lemma \ref{lem:SemInt} we would infer that (i) there is $y \geqslant x$ such that $y \vTrue A$ and $y \not\vTrue B$ and (ii) there is $z \geqslant x$ such that $z \vTrue B$ and $z \not\vTrue A$.  By Condition \eqref{L:Conn} (since $\mathbf{L} \in \mathbf{L_{1(0)c}}$) we would infer that $y \leqslant z$  or $z \leqslant y$. If $y \leqslant z$, by Lemma \ref{lem:Heredity} we would conclude that $z \vTrue A$ and $z \not\vTrue A$ (contradiction!); and if $z \leqslant y$, by Lemma \ref{lem:Heredity} we would conclude that $z \vTrue B$ and $z \not\vTrue B$ (contradiction!).

    \noindent \eqref{L:Ack}. In one direction, we need to show that for all $x$ (i) $x \vTrue A$ only if $x \vTrue t \to A$; and (ii) $x \vFalse t \to A$ only if $x \vFalse A$. (i) It suffices to show that $x \vFalse A$ only if $x \vFalse t$. By $x \vTrue A$, $x \vFalse A$ and Lemma \ref{lem:FSufficiency} we infer that $x \in F$, by which we conclude that $x \vFalse t$. (ii) If $x \vFalse t \to A$ because $x \vTrue t$ and $x \vFalse A$ we are done. If $x \vFalse t \to A$ because $x \not\vTrue t \to A$, then $\mathbf{L} = \mathbf{B_{10(c)}}$ and either there is $y \geqslant x$ such that $y \vTrue t$ and $y \not\vTrue A$ or there is $y \geqslant x$ such that $y \vFalse A$ and $y \not\vFalse t$. In both cases we conclude that $x \not\vTrue A$ using Lemma \ref{lem:Heredity} (and Lemma \ref{lem:FSufficiency} in the second case), which implies by $\mathbf{L} = \mathbf{B_{10(c)}}$  and Lemma \ref{lem:Valuation} that $x \vFalse A$. The other direction is straightforward. We need to show that for all $x$ (i) $x \vTrue t \to A$ only if $x \vTrue A$; and (ii) $x \vFalse A$ only if $x \vFalse t \to A$. (i) By $x \vTrue t \to A$, $x \leqslant x$ and $x \vTrue t$ we conclude that $x \vTrue A$ for all $\mathbf{L}$. (ii) By $x \vFalse A$ and $x \vTrue t$ we conclude that $x \vTrue t \land \neg A$, which implies $x \vFalse t \to A$ for all $\mathbf{L}$.
\end{proof}


Before moving to completeness, we record exactly why Dunn's version of the chain axiom $(A \to B) \lor (B \to A)$ is invalid in $\mathbf{BN_{10c}}$.

\begin{proposition}\label{prop:DunnFix}
    $\not\models_{\mathcal{N}_{\mathbf{BN_{10c}}}} (A \to B) \lor (B \to A)$.
\end{proposition}
\begin{proof}
    Take a $\mathbf{BN_{10c}}$-model $\M_{\mathbf{BN_{10c}}}$ such that $K=\{x,y,z\}$, $x \leqslant y \leqslant z$, $F=\{z\}$ and $0 \in V(p,z)$, $0 \notin V(q,z)$, $1 \in V(p,y)$ and $1 \notin V(q,y)$. Note that $V$ is well-defined only because we allow $V(p,z)=\{1,0\}$ and $V(q,y)=\emptyset$. Then, (i) $y$ is such that $y \geqslant x$, $y \models^1 p$ and $y \not \vTrue q$, whence $y \not\vTrue p \to q$; and (ii) $z$ is such that $z \geqslant x$, $z \models^0 p$ and $z \not \vFalse q$, whence $z \not\vTrue q \to p$. We conclude by (i) and (ii) that $x \not\vTrue (p \to q)$ and $x \not\vTrue (q \to p)$, whence $x \not\vTrue (p \to q) \lor (q \to p)$.
\end{proof}

Completeness is established by a canonical model construction. To this aim, we introduce as usual $\mathbf{L}$-theories, i.e. sets of formulas closed under $\mathbf{L}$-derivability. It follows from the definition of theories that every $\mathbf{L}$-theory contains all theorems of $\mathbf{L}$.

\begin{definition}\label{def:Theories}
An \emph{$\mathbf{L}$-theory} is $\Gamma \subseteq \mathcal{L}_{\mathbf{L}}$ such that $\Gamma \vdash_{\mathbf{L}} A$ only if $A \in \Gamma$ (i.e.\ $\mathrm{Th}_{\mathbf{L}}(\Gamma) = \Gamma$). An $\mathbf{L}$-theory $\Gamma$ is \emph{prime} iff $A \lor B \in \Gamma$ only if $A \in \Gamma$ or $B \in \Gamma$. The only \emph{trivial} $\mathbf{L}$-theory is $\mathcal{L}_{\mathbf{L}}$; an $\mathbf{L}$-theory $\Gamma$ is \emph{non-trivial} iff it is not trivial. 
\end{definition}

\noindent Note that if $\Gamma$ is a prime $\mathbf{L}$-theory, then $A \land B \in \Gamma$ iff $A \in \Gamma$ and $B \in \Gamma$ and $A \lor B \in \Gamma$ iff $A \in \Gamma$ or $B \in \Gamma$. For a non-empty $\Delta$, we write $\Gamma \vdash_{\mathbf{L}} \Delta$ iff $\Gamma \vdash_{\mathbf{L}} A_1 \lor \ldots \lor A_n$ for some non-empty subset $\{ A_1, \ldots, A_n \}$ of $\Delta$.

\begin{lemma}[Lindenbaum's Lemma]\label{lem:Lindenbaum}
If $\Gamma$ is an $\mathbf{L}$-theory and we have a non-empty $\Delta \subseteq \mathcal{L}_{\mathbf{L}}$ such that $\Gamma \not\vdash_{\mathbf{L}} \Delta$, then there is a non-trivial prime $\mathbf{L}$-theory $\Gamma'$ such that $\Gamma \subseteq \Gamma'$ and $\Delta \cap \Gamma' = \emptyset$.
\end{lemma}
\begin{proof}
The reader may not need to be reminded about this, but this result follows from the fact that prime $\mathbf{L}$-theories correspond to prime filters on the Lindenbaum-Tarski algebra for $\mathbf{L}$, which is a distributive lattice, and the fact that disjoint filter-ideal pairs in distributive lattices can be separated by prime filters; see \cite[Lemma 8.6.2]{DunnHardegree2001}.
\end{proof}

\begin{definition}\label{def:PropCanModel}
The \emph{canonical $\mathbf{L}$-structure} is $\langle K_{\mathbf{L}},(F_{\mathbf{L}}), \subseteq, V_{\mathbf{L}} \rangle$ where
\begin{itemize}
\item $K_{\mathbf{L}}$ is the set of non-trivial prime $\mathbf{L}$-theories;
\item $\Gamma \in F_{\mathbf{L}}$ iff $\neg t \in \Gamma$;
\item $1 \in V_{\mathbf{L}}(p,x)$ iff $A \in \Gamma$;
\item $0 \in V_{\mathbf{L}}(p,x)$ iff $\neg A \in \Gamma$.
\end{itemize}
\end{definition}

\begin{lemma}\label{lem:CanonicalModel}
The canonical $\mathbf{L}$-structure is a generalised Nelson model in $\mathcal{N}_{\mathbf{L}}$.
\end{lemma}
\begin{proof}
It suffices to show that the canonical $\mathbf{L}$-structure satisfies the model conditions satisfied by $\mathcal{N}_{\mathbf{L}}$. We show the details only for Condition  \eqref{eq:Ack} (assuming $\mathbf{L} \in \mathbf{L_{10}}$) and Condition \eqref{eq:Conn} (assuming $\mathbf{L} \in \mathbf{L_{1(0)c}}$). The reader is referred to \cite{Dunn2000} for the remaining cases.

\noindent \eqref{eq:Ack}. Assume $p \in \Gamma$ and $\neg p \in \Gamma$. Hence, by Definition \ref{def:Theories} we have $p \land \neg p \in \Gamma$. Note that  by \eqref{L:Ack} (since $\mathbf{L} \in \mathbf{L_{10}}$) and Lemma \ref{lem:StrongToInt} we have that $\vdash_{\mathbf{L}} p \land \neg p \To_\mathbf{L} \neg t$. Hence, we conclude that $\neg t \in \Gamma$ iff $\Gamma \in F$.

\noindent \eqref{eq:Conn}. Assume by contradiction that $\Gamma \subseteq \Delta_1$ and $\Gamma \subseteq \Delta_2$ but neither $\Delta_1 \not\subseteq \Delta_2$ nor $\Delta_2 \not\subseteq \Delta_1$. Hence, there are $A \in \Delta_1$ and $B \in \Delta_2$ such that $A \not\in \Delta_2$ and $B \not\in \Delta_1$. By Definition \ref{def:Theories} and  \eqref{L:Conn} (since $\mathbf{L} \in \mathbf{L_{1(0)c}}$)  we have $(A \To_\mathbf{L} B) \lor (B \To_\mathbf{L} A) \in \Gamma$, which implies by primeness of $\Gamma$ that either $A \To_\mathbf{L} B \in \Gamma$ or $B \To_\mathbf{L} A \in \Gamma$. If $A \To_\mathbf{L} B \in \Gamma$, by $\Gamma \subseteq \Delta_1$ we would infer that $A \To_\mathbf{L} B \in \Delta_1$, which would imply by  $A \in \Delta_1$ and \eqref{L:ToAdj} that $(A \To_\mathbf{L} B) \land A \in \Delta_1$. We would conclude by $\vdash_{\mathbf{L}} ((A \To_\mathbf{L} B) \land A) \To_\mathbf{L} B$ (which uses Lemma \ref{lem:StrongToInt}) that $B \in \Delta_1$ (contradiction!). The case $B \To_\mathbf{L} A \in \Gamma$ is established similarly.
\end{proof}

In what follows, the canonical $\mathbf{L}$-structure will be referred to as the canonical $\mathbf{L}$-model and it will be denoted as $\mathfrak{M}_\mathbf{L}^c$.

\begin{lemma}\label{lem:SoundCompletCanonicalModel}
$A \in \Gamma$ iff $\Gamma \vDash^1 A$ in $\mathfrak{M}_{\mathbf{L}}^c$.
\end{lemma}
\begin{proof}
It suffices to show for all $\mathfrak{M}^c_{\mathbf{L}}$ that (i) $A \in \Gamma $ iff $\Gamma \vTrue A$ and (ii) $\neg A \in \Gamma $ iff $\Gamma \vFalse A$. The result is established by induction on the complexity of $B$. All cases are straightforward, except the case $A= B \to C$. We use elliptical modus ponens without mention in this proof. 

\noindent (i). 
From left to right, assume $B \to C \in \Gamma$. We distinguish cases. If $\mathbf{L} \in \mathbf{L_{1(c)}}$ then by the induction hypothesis it suffices to show that $\forall \Delta \supseteq \Gamma(B \in \Delta$ only if $C \in \Delta)$. This holds trivially by $B \To_\mathbf{L} C, B \vdash_{\mathbf{L}} C$ and Lemma \ref{lem:StrongToInt}. If $\mathbf{L} \in \mathbf{L_{10(c)}}$, then we need to additionally show that $B \to C \in \Gamma$ only if $\forall \Delta \supseteq \Gamma(\neg C \in \Delta$ only if $\neg B \in \Delta)$. This again holds trivially by  $B \To_{\mathbf{L}} C, \neg C \vdash_{\mathbf{L}} \neg B$ (which relies on the contraposition axiom \eqref{L:toContrap} of $\mathbf{L} \in \mathbf{L_{10(c)}}$). 
From right to left, assume $B \to C \notin \Gamma$. We rely on the easily provable fact (cf. \cite[p.31]{Dunn2000}) that for all $\mathbf{L}$:
\begin{equation}
\Gamma \vdash_{\mathbf{L}} B \to C \textit{ if } 
\begin{cases}
    \Gamma, B \vdash C \textit{ and } \Gamma, \neg C \vdash \neg B & \textit{if } \mathbf{L} \in \mathbf{L_{10(c)}}\\
    \Gamma, B \vdash C  & \textit{otherwise.} \tag{*} \label{eq:*}
\end{cases}    
\end{equation}
\noindent We distinguish cases. If $\mathbf{L} \in \mathbf{L_{1(c)}}$, then by \eqref{eq:*} we infer that $\Gamma, B \not\vdash_{\mathbf{L_{1(c)}}} C$, which implies by Lemma \ref{lem:Lindenbaum} that there is $\Delta \in K$ such that $\Delta \supseteq \Gamma \cup \{A\}$ and $B \notin \Delta$. We conclude by the induction hypothesis that $\Delta \vTrue A$ and $\Delta \not \vTrue B$. If $\mathbf{L} \in \mathbf{L_{10(c)}}$, then  by \eqref{eq:*} we infer that $\Gamma, B \not\vdash_{\mathbf{L_{10(c)}}} C$ or $\Gamma, \neg C \not\vdash_{\mathbf{L_{10(c)}}} \neg B$ and we reason similarly as above.

\noindent (ii). From left to right, assume $\neg (B \to C) \in \Gamma$. We distinguish cases. 
If $\mathbf{L} \in \{\mathbf{N_{1(0)}},$ $\mathbf{BN_{1(0)}}\}$, then by axiom \eqref{L:to0CounterEg} we infer that $B \land \neg C \in \Gamma$, whence $B \in \Gamma$ and  $\neg C \in \Gamma$. We conclude by the induction hypothesis that $\Gamma \vTrue B \land \neg C$.
If $\mathbf{L} \in \mathbf{B_{1(0c)}}$, then by the induction hypothesis it suffices to show that  $B \land \neg C \in \Gamma$ or $B \to C \notin \Gamma$. To this aim, assume $B \to C \in \Gamma$ to show $B \land \neg C \in \Gamma$. By $B \to C \in \Gamma$ and $\neg (B \to C) \in \Gamma$ we infer that $(B \to C) \land \neg (B \to C) \in \Gamma$. We conclude by axiom \eqref{L:to0+NonTrue.1} of $\mathbf{L} \in \mathbf{B_{1(c)}}$ that $B \land \neg C \in \Gamma$, whence again $\Gamma \vTrue B \land \neg C$.
From right to left, we distinguish cases. If $\mathbf{L} \in \{\mathbf{N_{1(0c)}}, \mathbf{BN_{1(0c)}}\}$, then it suffices by the induction hypothesis to show that  $B \land \neg C \in \Gamma$ only if $\neg (B \to C) \in \Gamma$. This follows from axiom \eqref{L:to0CounterEg2}. If $\mathbf{L} \in \mathbf{B_{1(0c)}}$, then we need to additionally show that $B \to C \notin \Gamma$ only if $\neg (B \to C) \in \Gamma$. This follows from axiom \eqref{L:NoGaps} of $\mathbf{B_{1(c)}}$, primeness of $\Gamma$ and  $\mathbf{L} \subseteq \Gamma$.
\end{proof}

\begin{theorem}[Completeness]
    $\Gamma \models_{\mathcal{N}_\mathbf{L}} A$ only if $\Gamma \vdash_{\mathbf{L}} A$.
\end{theorem}
\begin{proof}
    It follows from Lemmas \ref{lem:CanonicalModel}  and  \ref{lem:SoundCompletCanonicalModel}.
\end{proof}

\section{Modal extensions}\label{sec:Modal}

Modal extensions of generalised Nelson logics are formulated in the languages $\mathcal{L}_{\Box \Diamond}$ and $\mathcal{L}^t_{\Box \Diamond}$ (the latter whenever $\mathbf{L} \in \{\mathbf{B_{10(c)}},\mathbf{BN_{10(c)}}\}$). The main contribution of this section is to show that, given an arbitrary generalised Nelson logic $\mathbf{L}$, its modal extension $\mathbf{L}_\Box$ can be obtained by uniformly extending the Hilbert system for $\mathbf{L}$ and the class of models $\mathcal{N}_{\mathbf{L}}$. To be precise, Hilbert systems for modal extensions of $\mathbf{L}$ are uniform only at the level of formula templates, since some axioms employ elliptical implication $\To_{\mathbf{L}}$, and we have seen that $\To_{\mathbf{L}}$ corresponds to different formulas in each $\mathbf{L}$.

\subsection{Hilbert systems}

Hilbert systems for modal generalised Nelson logics are in close relationship with axiom systems used for a number of non-classical modal logics, chiefly modal Nelson logics \cite{OdintsovWansing2004,Sherkhonov2008}, modal extensions of $\mathbf{FDE}$ \cite{Drobyshevich2020,Odintsov2010-ODIMLW-2,OdintsovWansing2017}, modal intuitionistic logics \cite{Bozic1984-BOIMFN,Simpson1994} and positive modal logics \cite{Dunn1995}.

The easiest way to present them is in comparison with positive modal logic. In \cite{Dunn1995}, Dunn axiomatises the minimal positive normal modal logic $\mathbf{K}^+$, over a language including $\land,\lor$ and the two independent modal operators $\Box$ and $\Diamond$. The basic modal principles are (i) monotonicity rules for $\Box,\Diamond$, (ii) distribution of $\Box$ over $\land$ and of $\Diamond$ over $\lor$, (iii) two interaction principles between $\Box$ and $\Diamond$, and (iv) the necessitation rule. The natural analogues of such principles are all theorems of each modal generalised Nelson logic $\mathbf{L}_\Box$, while the remaining principles relate $\Box$ and $\Diamond$ with strong negation $\neg$ and strong implication $\to$.

Starting from strong negation, we may note that, just like in $\mathbf{FDE}$, $\neg$ is a De Morgan negation, hence $\land$ and $\lor$ are dual. For this reason, as well as for the sake of simplicity, we decided to take dual $\Box$ and $\Diamond$ as well.\footnote{Note that this is in contrast with the famous intuitionistic requirement, formulated by Simpson \cite{Simpson1994}, that $\Box$ and $\Diamond$ be non-dual, to match the non-duality of $\land$ and $\lor$, and of $\forall$ and $\exists$. This requirement seems less pressing in the context of strong negation, as opposed to intuitionistic negation. However, failures of semantic and formal duality, in the style of \cite{OdintsovWansing2004,Sherkhonov2008,Drobyshevich2020}, are left for future work.} As a result, half of the axioms from positive modal logics turn out redundant. Next, modal generalised Nelson logics are akin to Nelson logics in that negation is a non-congruential operator, thus the logic is closed only under weak replacement of strongly equivalent formulas (cf. Lemma \ref{thm:PosNegReplacement}). To lift this result to a modal setting, we need to include in the axiomatisation the axioms $\Box \neg A \tot \neg \Diamond A$ and $\Diamond \neg A \tot \neg \Box A$, from which duality follows. Of course, many logics in the generalised Nelson family are congruential, in which case each of the previous axioms is derivable from the other. 

Moving to strong implication, we turn the reader's attention again to Lemma \ref{lem:StrongToInt}, stating that $\to$ is stronger than elliptical implication $\To_\mathbf{L}$ in each $\mathbf{L}$. For this reason, it suffices to formulate all axioms of modal generalised Nelson logics with $\to$ as the preferred outer implication. This is the case for all axioms except the $\mathbf{K}$-axioms for $\Box$ and $\Diamond$, where the presence of nested implications requires one to use elliptical implication $\To_\mathbf{L}$ rather than strong implication $\to$.\footnote{We note that for all generalised Nelson logics $\mathbf{L}$ except $\mathbf{BN_{10(c)}}$ it is sufficient to take the standard $\mathbf{K}$-axioms $\Box (A \to B)  \to (\Box A \to \Box B)$ and $\Box (A \to B)  \to (\Diamond A \to \Diamond B)$. This is because the versions of the axioms with $\To_{\mathbf{L}}$ are provable from their versions with $\to$. However, we were not able to show the latter fact for the case $\mathbf{L}=\mathbf{BN_{10(c)}}$.}  As a result,  modal generalised Nelson logics are akin to the positive fragment of Fischer--Servi's modal intuitionistic logic $\mathbf{IK}^+$.\footnote{A notable difference with $\mathbf{IK}^+$ (cf. \cite[p.52]{Simpson1994}) is the latter's interaction axiom $(\Diamond A \to \Box B) \to \Box (A \to B)$ (where $\to$ is intuitionistic implication), which is traded in the present context for $\Box(A \lor B) \to (\Box A \lor \Diamond B)$ (where $\to$ is strong implication).}

\begin{definition}[Modal Hilbert systems]
    For each generalised Nelson logic $\mathbf{L}$, its modal extension $\mathbf{L}_\Box$ is the smallest set of formulas in $\mathcal{L}_{\mathbf{L}}$ (i) containing all axioms of $\mathbf{L}$ as well as all formulas of the form
    \begin{gather}
  \Box (A \To_{\mathbf{L}} B)  \To_{\mathbf{L}} (\Box A \To_{\mathbf{L}} \Box B) \label{Lm:K}\\
 \Box (A \To_{\mathbf{L}} B)  \To_{\mathbf{L}} (\Diamond A \To_{\mathbf{L}} \Diamond B) \label{Lm:KDiam}\\
  \Diamond (A \lor B)  \to (\Diamond A \lor \Diamond B) \label{Lm:DiamDis}\\
 \Box (A \lor B)  \to (\Box A \lor \Diamond B)\label{Lm:DunnBox}\\
 \Box \neg A  \tot \neg \Diamond A\label{Lm:BoxNeg}\\
 \Diamond \neg A  \tot \neg \Box A\label{Lm:DiaNeg}
    \end{gather}
\noindent and (ii) closed under elliptical modus ponens \eqref{L:ToModusPonens} and
\begin{gather}
    \dfrac{A}{\Box A} \label{Lm:Nec}
\end{gather}
\end{definition}

An important remark concerns the definition of derivability in $\mathbf{L}_\Box$. To preserve the deduction theorem of Lemma \ref{thm:DT} in a modal setting, it is necessary to employ exactly the notion specified by Definition \ref{def:Lderivable}, which in particular rules out the application of the necessitation rule \eqref{Lm:Nec} to arbitrary assumptions but allows its application only to theorems. This is the so-called \emph{local} notion of derivability in modal logic.



The following lemma, which relies on axioms \eqref{Lm:K} and \eqref{Lm:KDiam}, is heavily and uniformly used in the completeness proof. 

\begin{lemma}\label{lem:CanonicalModel-modal-Help}
For all $\mathbf{L}_\Box$, the following hold:
\begin{align}
    &\dfrac{\Gamma \vdash_{\mathbf{L}_\Box} A}{\Box \Gamma \vdash_{\mathbf{L}_\Box} \Box A} \label{HelpBox}\\
&\dfrac{\Gamma, A \vdash_{\mathbf{L}_\Box} B}{\Box \Gamma, \Diamond A \vdash_{\mathbf{L}_\Box} \Diamond B} \label{HelpDiamond}
\end{align}
\noindent where $\Box \Gamma := \{ \Box A \mid A \in \Gamma \}$.
\end{lemma}
\begin{proof} 
We write $\To$ for $\To_{\mathbf{L}}$ in this proof. We observe that $(*)$ if $A \in \Gamma \cup \mathbf{L}_\Box$, then $\Box \Gamma \vdash \Box A$ thanks to \eqref{Lm:Nec}. 

\noindent \eqref{HelpBox}. If $\Gamma \vdash A$, then there is a sequence of formulas $\langle D_1, \ldots, D_n\rangle$ such that $D_n = A$ and, for all $i \leq n$, $D_i$ satisfies one of the following: (i) $D_i \in \mathbf{L}_\Box$, (ii) $D_i \in \Gamma$, or (iii) there are $j, k < i$ such that $D_k = D_j \To D_i$. We prove that $\Box \Gamma \vdash \Box D_i$ for all $i \leq n$. 
For $i = 1$, there are two options; either $D_1 \in \Gamma$ or $D_1 \in \mathbf{L}_\Box$. It follows by our observation $(*)$  that $\Box \Gamma \vdash \Box A$. 
Now assume that the claim holds for all $i < m$ for a given $m \leq n$. We prove the claim for $m$. There are three cases to consider, corresponding to (i)--(iii) of Definition \ref{def:Lderivable}. The cases (i) and (ii) are established using our observation $(*)$. 
To deal with case (iii) assume that $D_k = D_j \To D_m$ for some $j, k < m$. Note that we have $\Gamma \vdash D_j \To D_m$ and $\Gamma \vdash D_j$. By the induction hypothesis, $\Box \Gamma \vdash \Box (D_j \To D_m)$ and $\Box \Gamma \vdash \Box D_j$. By \eqref{Lm:K} we have $\Box \Gamma \vdash \Box (D_j \To D_m) \To (\Box D_j \To \Box D_m)$ and so $\Box \Gamma \vdash \Box D_m$ by two applications of \eqref{L:ToModusPonens}. This concludes the proof of the first claim.

\noindent \eqref{HelpDiamond}. If $\Gamma, A \vdash B$, then $\Gamma \vdash A \To B$ by Lemma \ref{thm:DT}. Consequently, $\Box \Gamma \vdash \Box (A \To B)$ by the first claim. By \eqref{Lm:KDiam} we have $\Box \Gamma \vdash \Box (A \To B) \To (\Diamond A \To \Diamond B)$. We obtain $\Box\Gamma \vdash \Diamond A \To \Diamond B$ using \eqref{L:ToModusPonens} and this entails $\Box \Gamma, \Diamond A \vdash \Diamond B$ by Lemma \ref{thm:DT}. 
\end{proof}


To conclude this section, note that a weak replacement theorem holds for modal generalised Nelson logics. In fact, many logics in the family are closed under replacement of strongly equivalent formulas, but it is well-known that Nelson logics and their modal extensions only satisfy the following weaker property.

\begin{lemma}\label{thm:PosNegReplacement}
 Let $A$ be an arbitrary formula be any formula built from $\{ p_1, \ldots, p_n \} $ using operators $\neg, \land, \lor, \Box$ and $\Diamond$. Let $D$ and $D'$ be arbitrary formulas. Then, for all $\mathbf{L}_\Box$, the following rule is admissible:
\begin{gather}
    \dfrac{D \tot D' \qquad \neg D \tot \neg D'}{A(p/D) \Tot_\mathbf{L} A(p/D')} \label{Lm:Replacement}
\end{gather}
\noindent where $A(p/B)$ is the result of substituting $B$ for $p$ in $A$.
 \end{lemma}
 \begin{proof} 
By induction on the complexity of $A$. We write $A(B)$ instead of $A(p/B)$, and $\Tot$ instead of $\Tot_\mathbf{L}$ in this proof.  
The base case is established as follows. If $A \neq p$, then $A(D) = A = A(D')$ and we are done. If $A = p$, then $A(D) = D$ and $A(D') = D'$ and so the claim follows directly from the assumption and Lemma \ref{lem:StrongToInt}. We show only the induction step where $A = \Box B$. By the induction hypothesis, we have $\vdash B(D) \Tot B(D')$ and $\vdash \neg B (D) \Tot \neg B(D')$. Since it can be easily shown, using \eqref{Lm:K}, \eqref{Lm:KDiam} and \eqref{Lm:Nec}, that $\mathbf{L}_\Box$ is closed under 
 \[ \dfrac{A \Tot B}{\Box A \Tot \Box B} \qquad\text{and}\qquad
 \dfrac{A \Tot B}{\Diamond A \Tot \Diamond B} \, ,\]
\noindent we obtain $\vdash \Box B(D) \Tot \Box B(D')$ and $\vdash \Diamond \neg B(D) \Tot \Diamond \neg B(D')$. The latter implies $\vdash \neg \Box B(D) \Tot \neg \Box B(D')$ using \eqref{Lm:DiaNeg} and Lemma \ref{lem:StrongToInt}. Since $(\Box B)(E) = \Box B(E)$, we obtain $\vdash A(D) \Tot A(D')$ and $\vdash \neg A(D) \Tot \neg A(D')$. This concludes the proof.
 \end{proof}

\subsection{Kripke models}

Our presentation of Kripke models for modal generalised Nelson logics aims to highlight their relationship with models for modal intuitionistic logic. As noted in Section \ref{sec:Intro}, one motivation for the present level of generality comes from the realisation that modal extensions of generalised Nelson logics can be carried out somewhat uniformly, namely with uniform axiom and rule templates,  once it is noted that all generalised Nelson logics contain $\mathbf{IPC}^+$. This uniformity is evident at the model theoretic level: to obtain modal generalised Nelson models, it is sufficient to (i) add to generalised Nelson models a binary, Kripke--style modal accessibility relation $R$; and (ii) to impose well-known confluence conditions relating $R$ and the partial order $\leqslant$, in the style of modal intuitionistic models \cite{Bozic1984-BOIMFN,Simpson1994}.  

Finally,  whenever $\mathbf{L} \in \{\mathbf{B_{10(c)}},\mathbf{BN_{10(c)}}\}$, we assume the new Condition \eqref{eq:RF}, which relates the set of inconsistent states $F$ and the modal accessibility relation.

\begin{definition}[Modal Models]\label{def:ModalGenNelsonModel}
    A \emph{modal generalised Nelson $\mathbf{L}_\Box$-model}, or simply $\mathbf{L}_\Box$-model, is a structure $\M_{\mathbf{L}_\Box} = \langle \M_{\mathbf{L}},R\rangle$ (for $\mathbf{L}$ from Table \ref{tab:DunnLogics}) where:
\begin{itemize}
    \item $\M_{\mathbf{L}} = \langle K, \leqslant,F, V\rangle$ is a generalised Nelson $\mathbf{L}$-model in $\mathcal{N}_{\mathbf{L}}$;
    \item $R \subseteq K^2$ is a binary relation such that:
    \begin{gather}
    x \leqslant y, \ yRz\text{  only if  } \exists w(xRw, w \leqslant z)  \label{eq:Confluence1}\\
    x \leqslant y, \ xRz \text{ only if } \exists w(yRw, z \leqslant w)
    \label{eq:Confluence2}\\
    y \in F, Rxy \text{ only if } x \in F \label{eq:RF}
    \end{gather}
    \item $V$ is recursively extended to the whole language $\mathcal{L}_{\Box\Diamond}$ ($\mathcal{L}^t_{\Box\Diamond}$ for $\mathbf{L} \in \mathbf{L_{10}}$) by updating Definition \ref{def:GenNelsonModel} as follows:
    
\smallskip
\begin{tabular}{lcllcl}
   $x \vTrue \Box A    $ & iff & $\forall y \in R(x)(y \vTrue A)$\hspace{5mm} & $x \vFalse \Box A   $ & iff & $\exists y \in R(x)(y \vFalse A)$\\
   $x \vTrue \Diamond A $ & iff & $ \exists y \in R(x)(y \vTrue A)$ & $x \vFalse \Diamond A $ & iff & $ \forall y\in R(x)( y \vFalse A)$
\end{tabular}
\smallskip
\end{itemize}
\noindent The class $\mathcal{N}_{\mathbf{L}_\Box}$ consists of all modal generalised Nelson $\mathbf{L}_\Box$-models $\M_{\mathbf{L}_\Box}$.
\end{definition}

The confluence conditions relate  $R$ and $\leqslant$, in such a way to lift hereditariness of verification and falsification along the partial order $\leqslant$. 
It is worth noting that, in the context of intuitionistic modal logic, the confluence conditions are introduced with the same aim but sometimes have a different formulation.\footnote{For example, the difference in confluence conditions in the logic $\mathbf{IK}$ of \cite{Simpson1994} is due to a different verification clause for $\Box$, which is used to break its duality with $\Diamond$. }

We note that our truth and falsity conditions for $\Box$ and $\Diamond$ are those of Odintsov and Wansing \cite{Odintsov2010-ODIMLW-2}, but there they are combined with a different Nelson-type propositional semantics. Our approach can therefore be seen as an extension of the Odintsov-Wansing approach to generalised Nelson logics.

Just like in the propositional setting, the proof of soundness is streamlined by the following three lemmas, which lift to a modal setting Lemmas \ref{lem:Valuation}, \ref{lem:Heredity} and \ref{lem:FSufficiency}. We note that the new Condition  \eqref{eq:RF} is required for Lemma  \ref{lem:ModalFSufficiency}, and that this is precisely the reason why, differently from \cite{Dunn2000}, we need to explicitly employ the set $F$ in the semantics for some modal generalised Nelson models.

\begin{lemma}\label{lem:ModalValuation}
    For all $\mathbf{L}_\Box$-models $\M_{\mathbf{L}_\Box}$:
    \begin{enumerate}
        \item If $\mathbf{L} = \mathbf{B_{1(0c)}}$, then $x \not\vTrue A$ only if  $x \vFalse A$;
        \item If $\mathbf{L} = \mathbf{N_{1(0c)}}$ then $x \vTrue A$ only if  $x \not\vFalse A$.
    \end{enumerate}
\end{lemma}
\begin{proof}
    By induction on the complexity of $A$. By Lemma \ref{lem:Valuation}, it suffices to show the induction step case $A = \Box B$, which follows by straightforward application of the induction hypothesis.
\end{proof}

\begin{lemma}\label{lem:ModalHeredity}
    For all $\mathbf{L}_\Box$-models $\M_{\mathbf{L}_\Box}$:
    \begin{enumerate}
        \item $x \leqslant y, x \vTrue A \text{ only if } y \vTrue A$;
        \item $x \leqslant y, x \vFalse A \text{ only if } y \vFalse A$.
    \end{enumerate}
\end{lemma}
\begin{proof}
    By double induction on the complexity of $A$. By Lemma \ref{lem:Heredity}, it suffices to show the induction step case $A = \Box B$. The case $A = \Diamond B$ is established similarly.
    
    \noindent \emph{1.} Assume $x \vTrue \Box B$ and $x \leqslant y$. Then, assume for some arbitrary $z$ that $Ryz$ to show $z \vTrue B$. By $x \leqslant y$, $Ryz$ and Condition \eqref{eq:Confluence1} we infer that there is $w \leqslant z$ such that $Rxw$. By $x \vTrue \Box B$ and $Rxw$ we infer that $w \vTrue B$, by which we conclude by $w \leqslant z$ and Lemma \ref{lem:ModalHeredity} that $z \vTrue B$.

    \noindent \emph{2.} Assume $x \vFalse \Box B$ and $x \leqslant y$.  By $x \vFalse \Box B$ we infer that there is $z$ such that $Rxz$ and $z \vFalse B$. By $x \leqslant y$, $Rxz$ and Condition \eqref{eq:Confluence2} we infer that there is $w\geqslant z$ such that $Ryw$. By $z \vFalse B$, $z \leqslant w$ and Lemma \ref{lem:ModalHeredity} we infer that $w \vFalse B$, whence we conclude by $Ryw$ that $y \vFalse \Box B$.
\end{proof}

\begin{lemma}\label{lem:ModalFSufficiency}
    For all $\mathbf{L}_\Box$-models $\M_{\mathbf{L}_\Box}$ such that $\mathbf{L} \in  \mathbf{L_{10(c)}}$, if $x \vTrue A \land \neg A$ for some $A \in \mathcal{L}_{\mathbf{L}}$, then $x \in F$.
\end{lemma}
\begin{proof}
    By induction on the complexity of $A$. By Lemma \ref{lem:FSufficiency}, it suffices to show the induction step case $A= \Box B$. If $x \vTrue \Box B$ and $x \vFalse \Box B$, then there is $y$ such that $Rxy$ and both $y \vTrue B$, $y \vFalse B$. Hence, by the induction hypothesis we infer that $y \in F$. We conclude by $y \in F$, $Rxy$ and Condition \eqref{eq:RF} that $x \in F$.
\end{proof}

\subsection{Soundness and completeness}

This section contains the main technical result of the paper, i.e. a uniform strong soundness and completeness proof for modal generalised Nelson logics. Soundness is, as usual, the easier part of the proof.

\begin{theorem}[Modal soundness]\label{thm:ModSoundness}
     $\Gamma \vdash_{\mathbf{L}_\Box} A$ only if $\Gamma \models_{\mathcal{N}_{\mathbf{L}_\Box}} A$.
\end{theorem}
\begin{proof}
By induction on the length of  $\mathbf{L}_\Box$-derivations. By Theorem \ref{thm:Soundness}, it suffices to show that \eqref{Lm:K}--\eqref{Lm:DiaNeg} are valid and that \eqref{Lm:Nec} preserves validity in $\mathcal{N}_{\mathbf{L}_\Box}$. We show only some representative cases.

\noindent \eqref{Lm:KDiam}. We need to show that $x \vTrue \Box (A \To_{\mathbf{L}} B)$ only if $x \vTrue \Diamond A \To_{\mathbf{L}} \Diamond B$, for all $\mathbf{L}_\Box$. To this aim, it suffices by Lemma \ref{lem:SemInt} to assume $x \vTrue \Box (A \To_{\mathbf{L}} B)$ and show that $y \vTrue \Diamond B$ for some arbitrary $y \geqslant x$ such that $y \vTrue \Diamond A$. $y \vTrue \Diamond A$ implies that  there is $z$ such that $Ryz$ $z \vTrue A$. By $x \leqslant y$, $Ryz$ and Condition \eqref{eq:Confluence1} we infer that there is $w \leqslant z$ such that $Rxw$. By $x \vTrue \Box (A \To_{\mathbf{L}} B)$ and $Rxw$ we infer $w \vTrue A \To_{\mathbf{L}} B$, which implies by Lemma \ref{lem:SemInt}, $w \leqslant z$ and $z \vTrue A$ that $z \vTrue B$. We conclude by $Ryz$ that $y \vTrue \Diamond B$. 

\noindent \eqref{Lm:DunnBox}. 
If $\mathbf{L} \in \mathbf{L_{10(c)}}$, we need to show that (i)  $x \vTrue \Box (A \lor B)$ only if $x \vTrue \Box A \lor \Diamond B$ and that  (ii) $x \vFalse \Box A \lor \Diamond B$ only if $x \vFalse \Box (A \lor B)$. To show (i), assume $x \vTrue \Box (A \lor B)$ and $x \not\vTrue \Box A$ to show $x \vTrue \Diamond B$. By $x \not\vTrue \Box A$ we infer that there is $y$ such that $Rxy$ and $y \not\vTrue A$, which implies that $y \vTrue B$. Hence, we conclude by $Rxy$ that $x \vTrue \Diamond B$. To show (ii), assume $x \vFalse \Box A$ and  $x \vFalse \Diamond B$ to show $x \vFalse \Box (A \lor B)$. By $x \vFalse \Box A$ we infer that there is $y$ such that $Rxy$ and $y \vFalse A$. By $Rxy$ and $x \vFalse \Diamond B$ we infer that $y \vFalse B$. By $y \vFalse B$ and $y \vFalse A$ we infer $y \vFalse A \lor B$, which implies $x \vFalse \Box (A \lor B)$. If $\mathbf{L} \notin \mathbf{L_{10(c)}}$, then  $\mathbf{L} \in \mathbf{L_{1(c)}}$ and we need to show only (i), reasoning as above.

\noindent \eqref{Lm:BoxNeg}.  The following two chains of equivalences hold for all $\mathbf{L}_\Box$. (i) $x \vTrue \Box \neg A$ iff $y \vTrue \neg A$ for all $y$ such that $Rxy$ iff $y \vFalse A$ for all $y$ such that $Rxy$ iff $x \vFalse \Diamond A$ iff $x \vTrue \neg \Diamond A$. (ii) $x \vFalse \neg \Diamond A$ iff $x \vTrue \Diamond A$ iff there is $y$ such that $Rxy$ and $y \vTrue A$ iff there is $y$ such that $Rxy$ and $y \vFalse \neg A$ iff $x \vFalse \Box \neg A$. If $\mathbf{L} \in \mathbf{L_{10(c)}}$, then (i) and (ii) suffice to establish the validity of \eqref{Lm:BoxNeg} in $\mathbf{L}_\Box$, otherwise $\mathbf{L} \in \mathbf{L_{1(c)}}$ and (i) alone suffices. The validity of \eqref{Lm:DiaNeg} is established similarly.

\noindent \eqref{Lm:Nec}. For all $\mathbf{L}_\Box$, assume $x \vTrue A$ for all $x$. Then, we conclude that $y \vTrue A$ for all $y$ such that $Rxy$.
\end{proof}

Turning to completeness, we adopt the same canonical model construction strategy of Section \ref{sec:PropCompleteness}. In devising the canonical modal accessibility relation, we resort to the familiar modal Nelson and modal intuitionistic definition.

\begin{definition}\label{def:ModCanModel}
The \emph{canonical $\mathbf{L}_\Box$-structure} is $\langle K_{\mathbf{L}_\Box},(F_{\mathbf{L}_\Box}), \subseteq, R_{\mathbf{L}_\Box}, V_{\mathbf{L}_\Box} \rangle$ where
\begin{itemize}
\item $K_{\mathbf{L}}$, $F_{\mathbf{L}_\Box}$ and $V_{\mathbf{L}_\Box}$ are defined over $\mathbf{\mathbf{L}_\Box}$- theories as in Definition \ref{def:PropCanModel};
\item $R_{\mathbf{L}_\Box} \Gamma \Delta$ iff (i) $\Box A \in  \Gamma$ only if $A \in \Delta$ and (ii) $\neg A \in  \Delta$ only if $\neg \Box A \in \Gamma$.
\end{itemize}
\end{definition}

\begin{lemma}\label{lem:ModCanonicalModel}
The canonical $\mathbf{L}_\Box$-structure is a modal generalised Nelson model in $\mathcal{N}_{\mathbf{L}_\Box}$.
\end{lemma}
\begin{proof}
    By the argument in the proof of Lemma \ref{lem:CanonicalModel}, it suffices to show that \eqref{eq:Confluence1}--\eqref{eq:RF} are satisfied by the canonical $\mathbf{L}_\Box$-structure. We omit the subscript $\mathbf{L}_\Box$ in this proof. 

    \noindent \eqref{eq:Confluence1}. Assume that $\Gamma \subseteq \Delta$ and $R\Delta\Delta'$. Let $\Lambda^+ = \{A \mid \Box A \in \Gamma\}$ and $\Lambda^- = \{\neg C \mid \neg \Box C \notin \Gamma\} \cup \{B \mid B \notin \Delta'\}$. We will prove that $\Lambda^+ \nvdash \Lambda^-$. Assume, for the sake of contradiction, that this is not the case. Then, there is $\bigvee_{i=1}^n B_i \notin \Delta'$ and $ \bigvee_{j=1}^m \neg \Box C_j \notin \Gamma$ such that  $\Lambda^+ \vdash \bigvee_{i=1}^n B_i \vee \bigvee_{j=1}^m \neg C_j$. It follows (using Lemma \ref{lem:StrongToInt}) that:
    \begin{align*}
        &\Gamma \vdash\Box ( \bigvee\nolimits_{i=1}^n B_i \vee \bigvee\nolimits_{j=1}^m \neg C_j)& &\text{by \eqref{HelpBox} and definition of } \Lambda^+\\
        &\Gamma \vdash\Box \bigvee\nolimits_{i=1}^n B_i \vee  \Diamond\bigvee\nolimits_{j=1}^m \neg C_j& &\text{by \eqref{Lm:DunnBox}}\\
        &\Gamma \vdash\Box \bigvee\nolimits_{i=1}^n B_i \vee  \bigvee\nolimits_{j=1}^m \neg \Box C_j & &\text{by \eqref{Lm:DiamDis} and \eqref{Lm:DiaNeg}}
    \end{align*}
    \noindent Since $\Gamma$ is a prime $\mathbf{L}_\Box$-theory, either $\Box \bigvee_{i=1}^n B_i \in \Gamma$ or $\Diamond\bigvee_{j=1}^m \neg C_j \in \Gamma$.
    If $\Box \bigvee_{i=1}^n B_i \in \Gamma$, then $\Box \bigvee_{i=1}^n B_i \in \Delta$ by $\Gamma \subseteq \Delta$, and we would conclude that $\bigvee_{i=1}^n B_i \in \Delta'$ by  $R\Delta\Delta'$ (contradiction!).
    If $\bigvee_{j=1}^m \neg \Box C_j \notin \Gamma$, then we reach the desired contradiction. Hence, we conclude that $\Lambda^+ \nvdash \Lambda^-$. By Lindenbaum’s Lemma \ref{lem:Lindenbaum}, there is a prime $\mathbf{L}_\Box$-theory $\Gamma'$ such that $R\Gamma\Gamma'$ and $\Gamma' \subseteq \Delta'$.

    \noindent \eqref{eq:Confluence2}. Assume that $R\Gamma \Gamma'$ and $\Gamma \subseteq \Delta$. Let $\Lambda^+ = \{A \mid \Box A \in \Delta\}$ and $\Lambda^- = \{\neg C \mid \neg \Box C \notin \Delta\}$. 
    We will prove that $\Lambda^+, \Gamma' \nvdash \Lambda^-$. 
    Assume, for the sake of contradiction, that $\Lambda^+, \Gamma' \vdash \Lambda^-$. Then there is $\bigwedge^m_{i=1}B_i \in \Gamma'$  and $\bigvee_{j=1}^n \neg \Box C_j \notin  \Delta$ such that $\Lambda^+, \bigwedge^m_{i=1}B_i \vdash \bigvee_{j=1}^n \neg C_j$. By $\bigwedge^m_{i=1}B_i \in \Gamma'$ and $R\Gamma \Gamma'$ we infer (thanks to  \eqref{Lm:DiaNeg} and Lemma \ref{thm:PosNegReplacement}) that  $\Diamond \bigwedge^m_{i=1}B_i \in \Gamma \subseteq \Delta$. It follows (using Lemma \ref{lem:StrongToInt}) that:
    \begin{align*}
        &\Delta, \Diamond \bigwedge\nolimits^m_{i=1}B_i \vdash \Diamond \bigvee\nolimits_{j=1}^n \neg C_j& &\text{by \eqref{HelpDiamond} and definition of } \Lambda^+\\
        &\Delta \vdash \Diamond \bigvee\nolimits_{j=1}^n \neg C_j&  & \text{by }\Diamond \bigwedge\nolimits^m_{i=1}B_i \in \Delta \\
        &\Delta \vdash  \bigvee\nolimits_{j=1}^n \neg \Box C_j&  & \text{by \eqref{Lm:DiamDis} and  \eqref{Lm:DiaNeg}. }
    \end{align*}
\noindent Since $\Delta$ is a prime $\mathbf{L}_\Box$-theory, we would infer that $\bigvee\nolimits_{j=1}^n \neg \Box C_j \in \Delta$ (contradiction!). By Lindenbaum’s lemma, there is a prime $\mathbf{L}_\Box$-theory $\Delta'$ such that $R\Delta\Delta'$ and $\Gamma' \subseteq \Delta'$.

\noindent \eqref{eq:RF}. Recall that this condition needs to be checked only  when $\mathbf{L} \in \{\mathbf{B_{10(c)}},\mathbf{BN_{10(c)}}\}$, which contain \eqref{L:toContrap} as an axiom. If $\Delta \in F$ and $R \Gamma \Delta$, then the former implies that $\neg t \in \Delta$, which implies by the latter that $\neg \Box t \in \Gamma$. By \eqref{L:Ack}, \eqref{L:toContrap}, \eqref{Lm:Nec} and Lemma \ref{lem:StrongToInt} we infer that $\vdash \neg \Box t \To_\mathbf{L} \neg t$. Hence, we infer by  $\neg \Box t \in \Gamma$ that  $\neg t \in \Gamma$, by which we conclude that $\Gamma \in F$.
\end{proof}

In what follows, the canonical $\mathbf{L}_\Box$-structure will be referred to as the canonical $\mathbf{L}_\Box$-model and it will be denoted as $\mathfrak{M}_{\mathbf{L}_\Box}^c$.

\begin{lemma}\label{lem:SoundCompletModCanonicalModel}
 $A \in \Gamma$ iff $\Gamma \vDash_{\mathfrak{M}_{\mathbf{L}_\Box}^c} A$.
\end{lemma}
\begin{proof}
By the argument of Lemma \ref{lem:SoundCompletCanonicalModel}, it suffices to show for all $\mathfrak{M}^c_{\mathbf{L}_\Box}$ that (i) $\Box B \in \Gamma $ iff $\Gamma \vTrue \Box B$ and (ii) $\neg \Box B \in \Gamma $ iff $\Gamma \vFalse \Box B$. The induction step case involving $\Diamond B$ follows from this using \eqref{Lm:BoxNeg}, \eqref{Lm:DiaNeg} and Lemma \ref{thm:PosNegReplacement}.

\noindent (i). The left-to-right direction follows from the definition of $R$. For the right-to-left direction, assume that $\Box B \notin \Gamma$. Let $\Lambda^{+} = \{ C \mid \Box C \in \Gamma \}$ and $\Lambda^{-} = \{ B \} \cup \{ \neg D \mid \neg \Box D \notin \Gamma \}$. We prove that $\Lambda^{+} \not\vdash \Lambda^{-}$; Lindenbaum's Lemma \ref{lem:Lindenbaum} then entails that there is $\Gamma'$ such that $R \Gamma  \Gamma'$ and $B \notin \Gamma'$ as desired. 
Thus, assume to the contrary that $\Lambda^{+} \vdash \Lambda^{-}$. Then there is  $\bigvee_{i=1}^n \neg \Box D_i \notin \Gamma$ such that $\Lambda^{+} \vdash B \lor \bigvee_{i = 1}^{n} \neg D_i.$
It follows (by using Lemma \ref{lem:StrongToInt}) that:
\begin{align*}
& \Gamma \vdash \Box \left( B \lor \bigvee\nolimits_{i = 1}^{n} \neg D_i\right)
	&& \text{by Lemma \ref{lem:CanonicalModel-modal-Help} and definition of } \Lambda^+\\
& \Gamma \vdash \Box B \lor \Diamond \bigvee\nolimits_{i = 1}^{n} \neg D_i
	&& \text{by \eqref{Lm:DunnBox}}\\
& 	\Gamma \vdash \Box B \lor \bigvee\nolimits_{i = 1}^{n} \neg \Box D_i
	&& \text{by \eqref{Lm:DiamDis} and \eqref{Lm:DiaNeg}}.	
\end{align*}
Since $\Gamma$ is $\mathbf{L}$-theory, either $\Box B \in \Gamma$ (contradiction!) or $ \bigvee\nolimits_{i = 1}^{n} \neg \Box D_i \in \Gamma$ (contradiction!). Hence, we conclude that indeed $\Lambda^{+} \not\vdash \Lambda^{-}$.

\noindent (ii). The right-to-left direction follows from the definition of $R$. For the left-to-right direction, assume that $\neg \Box B \in \Gamma$. Let $\Lambda^{+} = \{ \neg B \} \cup \{ C \mid \Box C \in \Gamma \}$ and $\Lambda^{-} = \{ \neg D \mid \neg \Box D \notin \Gamma \}$.  We prove that $\Lambda^{+} \not\vdash \Lambda^{-}$; the Lindenbaum lemma then entails that there is $\Gamma'$ such that $\Gamma R \Gamma'$ and $\neg B \in \Gamma'$ as desired. 
Thus, assume to the contrary that $\Lambda^{+} \vdash \Lambda^{-}$. Then there is $\bigvee_{i=1}^n \neg \Box D_i \notin \Gamma$  such that
$\{ C \mid \Box C \in \Gamma \}, \neg B \vdash \bigvee_{i = 1}^{n} \neg D_i$. 
It follows (by using Lemma \ref{lem:StrongToInt}) that:
\begin{align*}
& \Gamma, \Diamond \neg B \vdash \Diamond \bigvee\nolimits_{i = 1}^{n} \neg D_i
	&& \text{by Lemma \ref{lem:CanonicalModel-modal-Help} }\\
& \Gamma, \neg \Box B \vdash \bigvee\nolimits_{i = 1}^{n} \neg \Box D_i
	&& \text{by \eqref{Lm:DiamDis} and \eqref{Lm:DiaNeg}} \\
& \Gamma \vdash \bigvee\nolimits_{i = 1}^{n} \neg \Box D_i
	&& \text{by } \neg \Box B \in \Gamma.	
\end{align*}  
\noindent Since $\Gamma$ is a theory, we have $\bigvee\nolimits_{i = 1}^{n}\neg \Box D_i \in \Gamma$ (contradiction!). Hence, we conclude that indeed $\Lambda^{+} \not\vdash \Lambda^{-}$.
\end{proof}

\begin{theorem}[Modal completeness]\label{thm:ModCompl}
    $\Gamma \models_{\mathcal{N}_{\mathbf{L}_\Box}} A$ only if $\Gamma \vdash_{\mathbf{L}_\Box} A$.
\end{theorem}
\begin{proof}
    It follows from Lemmas \ref{lem:ModCanonicalModel}  and  \ref{lem:SoundCompletModCanonicalModel}.
\end{proof}

\section{Related work}\label{sec:Related}

There is a fair amount of existing work on modal extensions of logics related to the Dunn-Belnap logic $\mathbf{FDE}$ and Nelson's logics $\mathbf{N4}$ and $\mathbf{N3}$. 
As mentioned in the introduction, the main difference between our contribution and such work is that our semantics is simpler and that our framework covers the whole generalised Nelson logics family. 
In this section, we comment on the differences between existing work and our framework in more detail.

Odintsov and Wansing \cite{Odintsov2010-ODIMLW-2} study $\mathbf{BK}$, a modal extension of the four-valued propositional logic $\mathbf{B4}$, i.e.~an extension of $\mathbf{FDE}$ with the weak implication of \cite{ArieliAvron1996} and the falsum constant $\bot$. It was later shown in \cite{OdintsovWansing2017} that $\mathbf{BK}$ is closely related  to the modal bilattice logic $\mathbf{MBL}$ of \cite{RivieccioEtAl2017} and to the  modal extension of $\mathbf{FDE}$  of \cite{Goble2006}, $\mathbf{KN4}$; e.g.~$\mathbf{KN4}$ is definitionally equivalent to a fragment of $\mathbf{BK}$. The semantics for $\Box$ and $\Diamond$ in $\mathbf{BK}$ is the same as ours. 
However, their models are not partially ordered and $\mathbf{B4}$ is not part of the generalised Nelson logic family. 

Odintsov and Wansing \cite{OdintsovWansing2004} study $\mathbf{NK}^{-}$, i.e. positive constructive modal logic $\mathbf{PK}$ with strong negation. Adding duality axioms $\neg \Box A \tot \Diamond \neg A$ and $\neg \Diamond A \tot \Box \neg A$ to $\mathbf{NK}^{-}$ gives rise to the logic $\mathbf{NK}$. Adding further interaction axioms between $\Box$ and $\Diamond$ 
\[ \Diamond (A \to B) \to (\Box A \to \Diamond B) \quad\text{and}\quad (\Diamond (A \to A) \to \Box B) \to \Box B \] 
\noindent 
to $\mathbf{NK}$ gives rise to the logic $\mathbf{NK}^{d}$. 
Our semantics is closely related to the semantics for $\mathbf{NK}^{d}$, since both are based on ordered sets and the semantics for $\mathbf{NK}^{d}$, unlike the one for $\mathbf{NK}$, uses only one accessibility relation $R$ for both modal operators $\Box$ and $\Diamond$. However, the verification clause for $\Box A$ and the falsification clause for $\Diamond A$ differ from ours. This difference leads to the failure in $\mathbf{NK}^{d}$  of our Dunn's interaction axiom $\Box (A \lor B) \to (\Box A \lor \Diamond B)$. 

Relatedly, Sherkhonov \cite{Sherkhonov2008} studies $\mathbf{DK}$, a modal extension  of $\mathbf{N4}$ with non-dual $\Box$ and $\Diamond$.  
Adding the interaction axioms displayed above to $\mathbf{DK}$ gives rise to the logic $\mathbf{DK}^d$, while further adding the interaction axioms
\[ \neg \Box \neg (\neg A \to \neg B) \to (\neg \Diamond A \to \neg \Box B) \quad\text{and}\quad (\neg \Box \neg (A \to A) \to \neg \Diamond B) \to \neg \Diamond B \]
\noindent to $\mathbf{DK}$ and $\mathbf{DK}^d$ gives rise to $\mathbf{DK}^{d\neg}$ and $\mathbf{DK}^{dd\neg}$, respectively. Similarly to \cite{OdintsovWansing2004}, our framework differs from  \cite{Sherkhonov2008} in the semantic clauses used for modal operators and in the shape of valid interaction axioms between $\Box$ and $\Diamond$.

Drobyshevich \cite{Drobyshevich2020} has formulated a general framework for modal logics based on $\mathbf{FDE}$, which subsumes the mentioned frameworks of \cite{Goble2006,OdintsovWansing2004,Odintsov2010-ODIMLW-2,OdintsovWansing2017,Sherkhonov2008}. The framework of \cite{Drobyshevich2020} is based on non-normal partially ordered models, where full modalities $\Box$ and $\Diamond$ are obtained by pairing partially defined modal operators.
Drobyshevich's framework does not cover modal extensions of some of the logics in the generalised Nelson family. Most notably, the framework does not cover $\mathbf{B_{10c}}= \mathbf{RM}$, which was the primary motivation for the present paper. This can be seen by noting that in \cite{Drobyshevich2020} the Ackermann constant $t$ is absent and, more importantly, the semantics does not allow one to consider what we called the escape disjunct in the falsification of strong implication. In order to accommodate the latter specific feature of $\mathbf{RM}$ in the framework of \cite{Drobyshevich2020}, one plausible strategy could be to employ partially defined operators also for implication, and define different strong implications by pairing such partial operators.


Finally, Gao, Kozhemiachenko and Olivetti \cite{Gao2025} study a paraconsistent counterpart $\mathbf{CN4K}$ of the constructive modal logic $\mathbf{CK}$ of \cite{WIJESEKERA1990271}. The logic $\mathbf{CN4K}$ extends $\mathbf{N4}$ with independent $\mathbf{K}$ axioms and rules for $\Box$ and $\Diamond$. The logic is very weak, since e.g.~no interaction principle is assumed between $\Box$ and $\Diamond$. Models for $\mathbf{CN4K}$ resemble those for $\mathbf{DK}$, in that four distinct accessibility relations are used to interpret the truth and falsity of $\Box$ and $\Diamond$. 
Our semantics is most closely related to the monorelational semantics for the extension $\mathbf{CN4K^1}$ of $\mathbf{CN4K}$. However, no condition is assumed in \cite{Gao2025}  on the interaction of modal accessibility relations and the partial order, which leads to different verification and falsification clauses for $\Box$ and $\Diamond$.

\section{Conclusion}\label{sec:Conclusion}

In this paper, we studied modal extensions of generalised Nelson logics \cite{Dunn2000}, based on a bilateral Kripke-style semantics. Our main technical result was a modular soundness and completeness result for modal generalised Nelson logics.

There are several interesting problems we would like to tackle in the future. Firstly, we will explore extensions of the framework studied in this paper to a contextual epistemic logic framework in the style of \cite{SedlarVigiani2022,SedlarVigiani2024}. In such a framework, some states in a model are `classical', meaning that non-classical propositional connectives such as $\to$ and $\neg$ behave classically in those states (i.e.~$\neg A$ means `not $A$' and $A \to B$ means `not $A$ or $B$'). This models agents who `reason non-classically' while being situated in classical possible worlds. 
Secondly, as discussed already by Dunn \cite{Dunn1976-DUNAKS}, there is an intriguing resemblance between linear models for $\mathbf{RM}$ and models of tense logic with linear time. We will explore this connection more formally, and in a modal setting. 
Thirdly, to better assess the applicability of our logics, we will study their computational properties such as decidability and complexity of their respective validity problems. Comparisons to classical modal logic will be especially interesting. 
Fourthly, we will generalise our modal extensions of generalised Nelson logics to cover non-dual and non-normal modal operators. 
Finally, it would be interesting to extend Drobyshevich's general framework \cite{Drobyshevich2020} to cover modal extensions of $\mathbf{RM}$.

\paragraph{AI declaration statement.} During the preparation of this work, the authors used Google Gemini (version 3.1 Pro) to conduct a technical and editorial review of the LaTeX manuscript. Specifically, the tool was used to verify the definitions and proofs, to search for related literature, to provide feedback on text clarity, and to suggest improvements to the writing style. It was not used to suggest the research direction undertaken here or to generate the mathematical results in the manuscript. The authors thoroughly reviewed, edited, and validated all outputs and take full responsibility for the final content and originality of the manuscript.

\paragraph{Acknowledgement.} This work was supported by the Czech Science Foundation grant no.\ 25-17958J. 
The second author also acknowledges funding from the Research Foundation – Flanders (FWO) under grant no.~1215326N  (Junior Postdoc Fellowship). Part of the work on the present article was carried out while he was affiliated with Ghent University.

\printbibliography

\end{document}